\documentclass[final,journal]{IEEEtran}
\usepackage{etex}
\usepackage[table]{xcolor}
\usepackage[abs]{overpic}
\usepackage{amsmath,graphicx,pstricks}
\usepackage{psfrag}
\usepackage{graphicx}
\usepackage{algorithm,algorithmic}
\usepackage{hhline}
\usepackage{textcomp}
\usepackage{cite}
\usepackage{multirow}
\usepackage{amssymb,amsthm}
\usepackage{mathtools}
\usepackage{color}
\usepackage{tikz}
\usetikzlibrary{shapes.geometric}
\usetikzlibrary{calc}
\usepackage{url}
\usetikzlibrary{patterns}
\usepackage{pgfplots}
\usepackage{latexsym}
\usepackage{slashed}
\usepackage{centernot}
\usepackage{epstopdf}
\usepackage{array}
\usepackage[caption=false,font=footnotesize]{subfig}
\newcommand{\bs}{\boldsymbol}
\newcommand{\bt}{\mathbf}

\DeclareMathOperator*{\maximize}{maximize}

\tikzset{
    triangle/.style={
        draw,
        shape border rotate=0,
        regular polygon,
        regular polygon sides=3,
        fill=black,
        inner sep=0,
        minimum size=3mm
    }
}
\makeatletter
\newcommand*{\shifttext}[2]{%
  \settowidth{\@tempdima}{#2}%
  \makebox[\@tempdima]{\hspace*{#1}#2}%
}
\makeatother

\newtheorem{theorem}{Theorem}[]
\newtheorem{lemma}{Lemma}[]
\newtheorem{proposition}{Proposition}[]
\newtheorem{corollary}{Corollary}[]

\newtheorem{definition}{Definition}[]

\newtheorem{example}{Example}[]

\newcommand{\I}{\mathcal{I}}
\newcommand{\J}{\mathcal{J}}
\newcommand{\K}{\mathcal{K}}
\newcommand{\F}{\mathcal{F}}

\newcommand{\huij}{\bt{\bar{u}}_{ij}}
\newcommand{\hvgk}{\bt{\bar{v}}_{gk}} 

\newcommand{\hul}{\bt{\bar{U}}_i} 

\newcommand{\hhgkl}{\bt{\bar{H}}_{(i,gk)}}
\newcommand{\hhgkla}{\bt{\bar{H}}_{(i,gk)}^{(1)}}
\newcommand{\hhgklb}{\bt{\bar{H}}_{(i,gk)}^{(2)}}
\newcommand{\hhgklc}{\bt{\bar{H}}_{(i,gk)}^{(3)}}
\newcommand{\hhgkld}{\bt{\bar{H}}_{(i,gk)}^{(4)}}
\newcommand{\fgkl}{\bt{F}_{i,gk}}

\begin{document}
\title{Role of Interference Alignment in Wireless Cellular Network Optimization}

\author{Gokul~Sridharan,~\IEEEmembership{Member,~IEEE,} Siyu~Liu,~\IEEEmembership{Member,~IEEE,}
        and~Wei~Yu,~\IEEEmembership{Fellow,~IEEE,}
\thanks{Gokul Sridharan and Siyu Liu were with the Edward S. Rogers Sr. Department of Electrical and Computer Engineering, University of Toronto, Toronto, Ontario M5S3G4, Canada. Gokul Sridharan is now at the Wireless Information Network Laboratory, Department of Electrical and Computer Engineering, Rutgers, The State University of New Jersey, North Brunswick, New Jersey, 08902, USA (e-mail: gokul@winlab.rutgers.edu).}
\thanks{Wei Yu is with the Edward S. Rogers Sr. Department of Electrical and Computer Engineering, University of Toronto, Toronto, Ontario M5S3G4, Canada (e-mails: \{siyu,weiyu\}@comm.utoronto.ca).}
\thanks{This work was supported by the Natural Sciences and Engineering Research
Council (NSERC) of Canada and BLiNQ Networks Inc., Kanata, Canada. 
The material in this paper has been presented in part at the Asilomar Conference on Signals, Systems, and Computers, Asilomar, USA, November 2015.
}
}

\maketitle

\begin{abstract}

The emergence of interference alignment (IA) as a degrees-of-freedom optimal strategy motivates the need to investigate whether IA can be leveraged to aid conventional network optimization algorithms that are only capable of finding locally optimal solutions. To test the usefulness of IA in this context, this paper proposes a two-stage optimization framework for the downlink of a $G$-cell multi-antenna network with $K$ users/cell. The first stage of the proposed framework focuses on nulling interference from a set of dominant interferers using IA, while the second stage optimizes transmit and receive beamformers to maximize a network-wide utility using the IA solution as the initial condition. Further, this paper establishes a set of new feasibility results for partial IA that can be used to guide the number of dominant interferers to be nulled in the first stage. Through simulations on specific topologies of a cluster of base-stations, it is observed that the impact of IA depends on the choice of the utility function and the presence of out-of-cluster interference. In the absence of out-of-cluster interference, the proposed framework outperforms straightforward optimization when maximizing the minimum rate, while providing marginal gains when maximizing sum-rate. However, the benefit of IA is greatly diminished in the presence of significant out-of-cluster interference.

\end{abstract}

\maketitle
\vspace{-1mm}
\section{Introduction}
\label{section_intro}

Interference coordination through the joint optimization of the transmission variables has emerged as a promising technique to address inter-cell interference in dense cellular networks. Efforts to develop algorithms for such a joint optimization have largely been divided into two separate domains: that of network  utility maximization (NUM) over power, beamforming and frequency allocation and that of  interference alignment (IA) for maximizing the degrees of freedom (DoF) of multi-antenna cellular  networks. The relationship between the two, however, remains largely unexplored. This paper attempts to answer the important and practically relevant question of whether or not DoF-focused IA algorithms can make an impact on wireless cellular network optimization. In particular, this paper focuses on the impact of IA in maximizing a given network utility in a  $G$-cell cluster having $K$ users/cell, with $N$ antennas at each  base-station (BS) and $M$ antennas at each user---a $(G,K,M\times N)$ cluster, with and without out-of-cluster interference.

\subsection{Motivation and Existing Work}

Joint optimization in coordinated cellular networks is an area of active research  \cite{gesbertjsac, venturino, weiyutwc, zakhour, minmaxCWT2, wmmse, hayssam, yangkim, minmaxCWT, minmaxyichao, minmaxraza, binda, newref1, newref2,newref3,newref4,newref5,newref6,newref7,newref8}. Typically, such an optimization problem involves the  maximization of a network-wide utility function (weighted-sum-rate, max-min-fairness rate, etc) over transmission parameters such as beamformers and transmit powers. Several novel techniques that exploit  equivalence relations between various problem formulations (e.g., weighted-sum-rate maximization and weighted mean-squared-error minimization  \cite{wmmse, minmaxraza}) or use concepts like uplink-downlink  duality \cite{hayssam,minmaxyichao,newref7} have been proposed in the context of NUM. However, irrespective  of the problem formulation and the proposed  solution, the non-convex nature of these problems makes it challenging to find efficient methods capable of finding solutions that are closer to the global optimum.

In parallel to these developments, significant progress has been made in establishing the DoF of multi-antenna cellular networks \cite{cadambejafar, maddah, chenweiwang,tingtingliu,razaviyayn,oscargonzalez,gokulArxiv,suh, suh2, huasunjafarisit2012,liu-yang-arxiv,moewin}. IA, with and  without symbol extensions, has played a key role in establishing these results \cite{cadambejafar, maddah, chenweiwang, zhuangfeasibility,yetis,razaviyayn,tingtingliu, gokulArxiv,oscargonzalez}. Since the capacity  of cellular networks is still unknown, DoF provides crucial insight on the limits of  cellular networks at high signal-to-noise ratios (SNRs). In particular, IA using beamforming techniques without symbol extensions has attracted significant attention due to its simple form and relative ease of implementation from a practical standpoint \cite{ zhuangfeasibility, yetis, chenweiwang, razaviyayn, tingtingliu,gokulArxiv,tresch, oscargonzalez, gomadam, peters2009}. While IA appears to be a reasonable goal to pursue, its impact on practical network performance cannot be taken for granted. Indeed, due  to the limited focus of IA on interference suppression while neglecting signal strength, IA cannot be viewed as a substitute for NUM and must instead be considered as a potential augmentation to the optimization process. Note that since both IA and NUM  algorithms place similar requirements on channel-state information (CSI) and thus have a similar overhead, it is pertinent to assess the value of IA in relation to NUM under realistic channel conditions that include pathloss, shadowing and fading.


Note that the goal of this paper is different from the many existing algorithms that minimize mean-squared-error (MSE) as a proxy for some network utility function \cite{wmmse,gomadam,schmidtTSP,schmidt2009,minmaxraza, christensen, peters2009}. These algorithms do not explicitly compute aligned beamformers but have been empirically observed to converge to aligned beamformers at high SNRs (i.e., high transmit powers). Although this observation appears to suggest that such algorithms implicitly account for the value of aligned beamformers, they do not explicitly compute or utilize aligned beamformers at finite SNRs and thus do not shed light on the value of IA at finite SNRs.

This paper tries to fill this void by examining whether NUM algorithms can benefit from explicit IA, even at finite SNRs. In particular, this paper investigates whether the network-utility landscape as a function of the beamformer coefficients for a fixed SNR is such that local minima close to aligned beamformers achieve higher network utility (as a consequence of better interference mitigation) than those obtained via a random initialization of the beamformers. Note that while the network-utility landscape changes as a function of SNR, the aligned beamformers are invariant to SNR. Thus, our goal is to examine whether this fixed neighbourhood around aligned beamformers carries any particular significance in the context of NUM at finite SNRs.

Significance of aligned beamformers in the context of NUM has been considered before \cite{taowang,guleryener,binda}, albeit in a limited context. In \cite{taowang}, it was observed that MSE minimization algorithms perform better when initialized to aligned beamformers in the 3-user interference channel. More recently, \cite{guleryener} considers using aligned beamformers in a heterogeneous network to null interference from macro users to femto BSs. In particular, when the macro users judiciously null interference to a selected set of femto BSs using aligned beamformers, it is shown that the overall macro-femto throughput increases. But an important issue that is not addressed by either of these papers is the fact that setting aside subspaces at transmitters and receivers for aligning interference comes at the cost of not using those dimensions to schedule more users or to deliver more spatial streams to a user. Thus, a fair assessment of the value of IA at finite SNRs must also consider cases where spatial multiplexing, i.e., maximizing total number of data streams, is prioritized, with no regard for IA.
Studying this trade-off is further warranted by a recent result in \cite{mungara}, where using tools from stochastic geometry it is established that IA (without any NUM) only rarely outperforms spatial multiplexing in large cellular networks where there is unavoidable out-of-cluster interference. To obtain a comprehensive insight on the value of IA at finite SNRs, this paper also considers the trade-off between spatial multiplexing and IA by varying the number of users scheduled in a time-frequency slot.

\subsection{Proposed Framework}

With the two-pronged goal of studying the value of IA in NUM at finite SNRs and the trade-off between spatial multiplexing and IA,  this paper proposes a two-stage optimization framework that is specifically designed to shed light on these issues. The proposed framework assumes perfect CSI to be available at a centralized location. Although the assumption  of perfect CSI for IA has come under significant scrutiny \cite{elayach1, elayach2}, given that the value of IA  has not been fully established in the NUM context even with perfect CSI, this paper makes the perfect CSI assumption and focuses solely on the role of IA in NUM.

As a first step in developing the proposed framework, this paper establishes certain crucial results on feasibility of partial IA via beamforming without symbol extensions. Partial IA refers to the selective nulling of interference from certain interferers in a given $(G,K,M\times N)$ cluster. It is well known that when complete IA is required in a $(G,K,M\times N)$ cluster while ensuring 1 DoF/user, $M+N \geq (GK+1)$ is a necessary and sufficient condition for the feasibility of designing transmit and receive beamformers for IA without symbol extensions \cite{zhuangfeasibility,yetis,razaviyayn,tingtingliu}. However, such a condition is too restrictive in networks with realistic channels where complete IA may not be feasible or even necessary. In particular, we establish necessary and sufficient conditions for feasibility of partial IA in a $(G,K,M\times N)$ cluster when interference from a set of $q$ BSs (dominant interferers) is nulled at each user.

We then switch focus to the design of the two-stage optimization framework. The first stage of this framework exclusively focuses on mitigating interference from the dominant interferers using IA. This stage relies significantly on the feasibility results of partial IA to guide the design choices. The second stage uses this altered interference landscape to optimize the network parameters to maximize a given utility function. Such a framework counters the  myopic nature of straightforward NUM algorithms by leveraging IA's ability to comprehensively address interference from the dominant interferers while subsequently relying on numerical optimization algorithms to account for signal strength and to maximize the network utility. Such a framework is uniquely suited to assess the impact of IA on NUM as it leverages IA's strengths on nulling interference to aid the performance of algorithms for NUM without altering their functioning in any significant manner. Note that recent work on multilevel topological interference management also advocates a similar approach to  manage interference in wireless networks \cite{multileveltopo}. Such an approach, proposed to achieve a certain number of generalized degrees of freedom, requires decomposing the network into two  components, one consisting of links that  correspond to interference that needs to be avoided or nulled, and the other consisting of links  where interference is sufficiently weak and is handled through power control. Our effort can be thought of in similar terms but in a more practical setting with the overall objective of maximizing  a utility function.

We use the two-stage framework to maximize either the minimum rate to the scheduled users (max-min fairness) or the sum-rate subject to per-BS power constraints. Using the results on the feasibility of partial IA in a  $(G,K,M\times N)$ cluster, in the first stage, we identify a requisite number of dominant interfering BSs to be nulled for each user. After aligning interference from the dominant BSs, we alternately optimize the transmit and receive beamformers to maximize the utility function.

Unlike typical studies on IA, we hold maximum transmit power fixed and instead use density (by decreasing inter-BS spacing) as a proxy for simulating networks that are severely interference-limited. With increasing density, the received signal strength and the observed interference grow at the same rate. Such a regime closely reflects practical network deployments and is different from other IA studies where the strength of uncoordinated interference is held fixed, while transmit power is increased to infinity\cite{peters2009}. The intention of our approach is to scan a wide range of scenarios from the severely-interference-limited regime to the noise-limited regime. IA is known to be DoF-optimal, but its value at moderate to weak SINRs is not very clear. Increasing network density allows us to explore the question on how much interference the network can tolerate before alignment becomes indispensable.

\subsection{Key Observations}
When maximizing the network utility of \emph{minimum rate} across the scheduled users in specific topologies of an \emph{isolated cluster} of BSs under realistic channel conditions, simulation results indicate that IA is helpful in the sense that (a) aligned beamformers do not naturally emerge from straightforward NUM algorithms even at high signal-to-noise  ratios; (b) aligned beamformers provide a significant advantage as initial condition to NUM, especially when BSs are  closely spaced; and (c) IA provides insights on the optimal number of users to be scheduled per cell. In particular, fewer number of users should be scheduled per cell as the BS-to-BS distance  decreases until the number of users per cell reaches $ K=\lfloor \frac{M+N-1}{G} \rfloor$.

These observations however do not apply when maximizing the network utility of sum-rate across the users. In particular it is seen that the impact of IA on sum-rate maximization for an isolated cluster of BSs is marginal. This draws attention to the choice of utility function in evaluating the role of IA. The key difference between the two utility functions is that when maximizing sum-rate, the number of users eventually served can change during the optimization process (some users may not be allocated any transmit power), and may not be the same as the number chosen by the original scheduler. Since aligned beamformers are designed based on the users originally scheduled by the scheduler, any subsequent change in the users served undermines the impact of IA. 

It is further seen that IA has no impact on maximizing either utility in the presence of out-of-cluster interference. In the presence of uncoordinated interference, cancelling interference from just one or two dominant BSs does not sufficiently affect the optimization landscape to yield better solutions. Although these observations are based on a particular optimization framework proposed in this paper, these results suggest that IA cannot be used indiscriminately to address all interference-related issues. Instead, it is advisable to adopt a more targeted approach to using IA, with its use being restricted to certain special cases.

\subsection{Paper Organization}

This paper is organized as follows. Section \ref{section_signal_model} establishes the system model used in this paper. Section \ref{section_PIA} discusses a set of feasibility results on partial IA, while Section \ref{section_optframework} describes the proposed optimization framework. The performance of the proposed optimization framework and the role of IA in NUM are discussed in Section \ref{section_simulations}.

\subsection{Notation}
In this paper, column vectors are represented in bold lower-case letters and matrices in bold upper-case letters. The conjugate transpose and Euclidean norm of a vector $\bt v$ are denoted as $\bt v^H$ and $\|\bt v \|$, respectively. The identity matrix is denoted as $\bt I$. Calligraphic letters (e.g., $\mathcal{Q}$) are used to denote sets and $|\cdot |$ is used to refer to the number of elements in a set. The notation $\mathcal{CN}(\bs \mu, \sigma^2 \bt I)$ is used to denote the multi-variate circularly-symmetric complex Gaussian density function with mean $\bs \mu$ and  variance $\sigma^2$ in each dimension.

\section{System Model}
\label{section_signal_model}
\begin{figure}
\centering
   \includegraphics[width=3.4in]{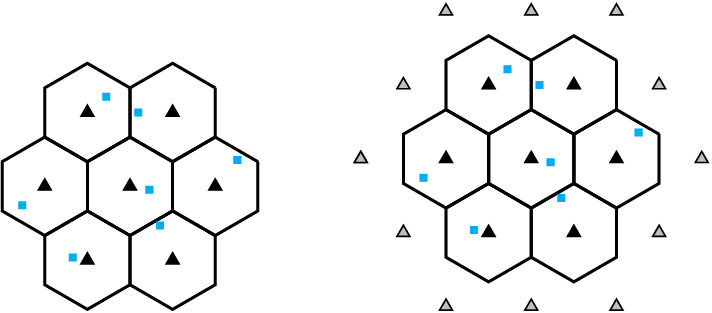}
 \caption{An isolated 7-cell cluster and a 7-cell cluster experiencing out-of-cluster interference. The squares represent users in each cell.}
 \label{fig_signal_model}
\end{figure}
Consider the downlink of a cellular network consisting of a cluster of $G$ interfering cells with $K$ users per cell. These $G$ interfering cells could either be isolated or be in the presence of several other interfering cells, resulting in out-of-cluster interference, as shown in Fig. \ref{fig_signal_model}. Each user is assumed to have $M$ antennas and each BS is assumed to have $N$ antennas. Let  the channel from the $i$th BS to the $k$th user in the $g$th cell be denoted  as the $M\times N$ matrix $\bt H_{(i,gk)}$. The channel model includes pathloss, shadowing and fading. We assume that the channel coefficients are known perfectly at a central location. Additionally, it is assumed that all fading coefficients are generic (or equivalently,  drawn from a continuous distribution). Assuming that each user is served with one data stream, the transmitted signal  corresponding to the $k$th user in the $g$th cell is given by $\bt{v}_{gk} s_{gk}$, where  $\bt v_{gk}$ is a $N \times 1$ linear transmit beamforming vector and $s_{gk}$ is the symbol to be  transmitted. This signal is received at the intended user using a $M \times 1$ receive  beamforming vector $\bt u_{gk}$. The received signal after  being processed by the receive beamforming vector can be written as 
\begin{align}
\bt{u}_{gk}^H\bt{y}_{gk}&=\sum_{i=1}^G\sum_{j=1}^K \bt{u}_{gk}^H \bt H_{(i,gk)}\bt{v}_{ij}{s}_{ij}+\bt{u}_{gk}^H\bt z_{gk}+\bt{u}_{gk}^H\bt n_{gk},
\label{eq_uhv}
\end{align}
where $\bt n_{gk}$ is the $M\times 1$ vector representing additive white Gaussian noise and $\bt z_{gk}$ represents cumulative out-of-cluster interference received at the $(g,k)$th user. A 7-cell cluster with and without out-of-cluster interference is represented in Fig.~\ref{fig_signal_model}. Since distance-dependent pathloss is included in the channel model, the cumulative interference at a user depends on the overall network topology and is a function of distance to the out-of-cluster interferers.  This paper restricts attention to beamforming based IA without symbol extensions in time or frequency. With the goal of evaluating the role of IA for NUM, we first investigate feasibility conditions for IA in cellular networks.

\section{Feasibility of Partial Interference Alignment} 
\label{section_PIA}
This section presents a set of conditions for the feasibility of IA when interference from only a subset of BSs within the cluster is cancelled at a user. Since interference from only a subset of  the interferers is aligned, we call this partial IA. It is important to  establish these results as complete IA may not be feasible in a given cluster  and sometimes, even unnecessary. We later use this result to guide how many and which interferers to align for NUM.

In the $G$-cell cluster described above, we construct a list $\mathcal{I}$ of  BS-user pairs where each pair indicates the need to cancel interference from a specific BS to a specific user. Let the double index $gk$ denote the $k$th user in the $g$th cell  and the single index $l$ denote the $l$th BS. For example, if the pair $(3,12) \in \mathcal{I}$,  this implies that the interference from the third BS  is to be completely nulled at the second user in the first cell. Satisfying this condition requires solving the following $K$ equations: 
\begin{align}
\bt u^H_{12} \bt H_{(3,12)} \bt v_{3j}= 0, \qquad \forall j \in \{1,2,\hdots, K \}.
\label{PIAeq}
\end{align}
In addition to these conditions, we also require the set of transmit beamformers at any BS to be linearly independent, i.e, rank$([\bt v_{g1},\bt v_{g2},\hdots, \bt v_{gK}])=K$.

Cancelling interference from only a subset of the interfering BSs is analogous to complete IA in partially connected cellular networks where certain cross links are assumed to be completely absent \cite{guillaudgesbert}. When the set $\mathcal{I}$ consists of all the $(G-1)GK$ possible pairs (denoted as $\mathcal{I}_{all}$), we get the usual conditions for complete IA \cite{zhuangfeasibility, yetis}. Each of the $K$ equations in (\ref{PIAeq}) is quadratic and collectively they form a polynomial system of equations. Feasibility of the system of polynomial equations when $\mathcal{I}=\mathcal{I}_{all}$ is well studied using tools from algebraic geometry \cite{razaviyayn,tingtingliu,bresler,oscargonzalez}. These tools have been applied to establish the feasibility of IA for the interference channel \cite{razaviyayn,oscargonzalez} and for the cellular network \cite{tingtingliu} under complete interference cancellation. The same set of tools can also be used to establish conditions for feasibility of partial IA for any given $\mathcal{I}$. The following theorem establishes one such result.
\begin{theorem}
\label{feasibilitytheorem}
Consider a $(G,K,M\times N)$ cluster where each user is served with one data stream. Let $\bt v_{gk}$ and $\bt u_{gk}$ denote the transmit and receive beamformer corresponding to the $(g,k)$th user where the set of beamformers $\{\bt v_{g1},\bt v_{g2},\hdots,\bt v_{gK} \}$ is linearly independent for every $g$. Further, let $\mathcal{I} \subseteq \{(i,gk):\, g\neq i,\, 1 \leq g,\, i\leq G,\, 1 \leq k \leq K \}$ be a set of user-BS pairs such that for each $(i,gk) \in \mathcal{I}$ the interference caused by the $i$th BS at the $(g,k)$th user is completely nulled, i.e.,
\begin{align}
\bt u^H_{gk} \bt H_{(i,gk)} \bt v_{ij}= 0, \qquad \forall j \in \{1,2,\hdots, K \}.
\end{align}
A set of transmit and receive beamformers $\{ \bt v_{gk} \}$ and $\{ \bt u_{gk}\}$ satisfying the polynomial system defined by $\mathcal{I}$ exists if and only if
\begin{gather}
 M\geq 1,\qquad  N \geq K,
 \label{signalineq}
\end{gather}
and
\begin{gather}
|\mathcal{J}_{users}|(M-1)+|\mathcal{J}_{BS}|(N-K )K  \geq |\mathcal{J}| K 
\label{complexcondition}
\end{gather}
for all $\mathcal{J} \subseteq \mathcal{I}$ where $\mathcal{J}_{users}$ and $\mathcal{J}_{BS}$ are the set of user and BS indices that appear in $\mathcal{J}$.
\end{theorem}

The proof of this theorem broadly follows the technique used in \cite{razaviyayn, tingtingliu} and is presented in Appendix \ref{algebraicgeometryproof}. Note that the inequalities in (\ref{signalineq}) are needed to ensure that the users and BSs have the minimum necessary antennas to decode the desired signal streams. Although intra-cell interference cancellation is not explicitly mentioned in the theorem, it can be subsequently eliminated through a simple linear transformation of the linearly independent transmit beamformers in each cell. A useful corollary that emerges from this theorem is stated below.
\begin{corollary}
\label{cor1}
Suppose the set $\mathcal{I}$ is such that each user in a $(G,K,M\times N)$ cluster requires interference from no more than $q$ BSs to be cancelled, where $1\leq q \leq G-1$, and each BS has no more than $Kq$ users that require this BS's transmission to be nulled at these users, then a set of sufficient conditions for the feasibility of IA is given by
\begin{gather}
 M\geq  1,\qquad  N \geq   K,
\end{gather}
and
\begin{gather}
M+N \geq  K(q+1)+1.
\label{simplecondition}
\end{gather}
\end{corollary}
\begin{proof}
First note that for any choice of $\mathcal{J}$ in (\ref{complexcondition}), the assumptions in the corollary imply $|\mathcal{J}|  \leq  \min \left(|\mathcal{J}_{BS}| Kq , |\mathcal{J}_{users}|q \right )$. Assume that $|\mathcal{J}_{BS}| Kq \geq |\mathcal{J}_{users}|q $; the proof for the other case can be established in a similar manner. Now, the following inequalities show that if (\ref{simplecondition}) is true, then (\ref{complexcondition}) holds for any choice of $\mathcal{J}$:
\begin{align}
&|\mathcal{J}_{users}|(M-1)+|\mathcal{J}_{BS}|(N-K )K \nonumber \\
& \geq |\mathcal{J}_{users}|(M-1+N-K) \nonumber \\
& \geq |\mathcal{J}_{users}|(K(q+1)-K) \quad (\text{using (\ref{simplecondition})})\nonumber\\
& = |\mathcal{J}_{users}|Kq \nonumber \\
& \geq |\mathcal{J}|K. \nonumber 
\end{align}
This completes the proof.
\end{proof}

Note that when $q=G-1$, we recover the well-known proper-improper condition for MIMO cellular  networks \cite{zhuangfeasibility}, i.e.,
\begin{align}
M+N \gtrless GK+1.
\end{align}
Fig.~\ref{feasibilityTable} illustrates the conditions of Corollary \ref{cor1} imposed on a $(4,2,3\times 4)$ cluster for the feasibility of partial IA. In this case,
\begin{align}
q \leq \left \lfloor \frac{M+N-1}{K} \right \rfloor-1 =2.
\end{align}
Each entry in  Fig.~\ref{feasibilityTable} represents a BS-user pair as identified by its row and column  indices. If a certain BS-user pair is in $\mathcal{I}$, the corresponding entry is marked with a  `$\times$'. Corollary \ref{cor1}  requires  $\mathcal{I}$ to be such that each row has no more than $q$ chosen entries and each column has no more than $Kq$ chosen entries. In particular, Fig.~\ref{feasibilityTable} considers feasibility of partial IA in a $(4, 2, 3 \times 4)$ cluster where each user can request interference from no more than $q = 2$ BSs to be cancelled and each BS can null interference at no more than $Kq = 4$ out-of-cell users. It is easy to see that the set of BS-user pairs chosen for interference cancellation satisfy the conditions imposed by Corollary \ref{cor1}.

This corollary provides a simpler set of guidelines on  choosing the set of BS-user pairs $(\mathcal{I})$ for partial IA than Theorem  \ref{feasibilitytheorem} where the number of feasibility constraints grows exponentially with the size of $\mathcal{I}$. However, designing $\mathcal{I}$ according to  this corollary rather than Theorem \ref{feasibilitytheorem} comes at the cost of simplifying restrictions on $\mathcal{I}$ that may otherwise be unnecessary. Note also that a key assumption of the feasibility condition derived in this paper is that each user is served one data-stream. Generalization of this condition to the multi-data-stream-per-user case is difficult and is in fact still an open problem even for the fully-connected case\footnote{Although a numerical test to verify feasibility in the multi-stream case is provided in \cite{oscargonzalez}, a closed form characterization of feasibility is not yet available.} \cite{razaviyayn,oscargonzalez,tingtingliu}.

Restricting to the single data-stream case, a crucial observation from Corollary \ref{cor1} is that there exists a trade-off between $K$, the number of users served in each cell, and $q$, the number of interferers each user can cancel interference from. The direct and interfering channel strengths in practical cellular networks can vary significantly. Intuitively, a cellular network should require interference nulling from only the dominant interferers while serving as many users per cell as possible. This necessitates a careful design of the set $\mathcal{I}$ while ensuring feasibility of  partial IA. The condition in the corollary plays an important role in network optimization framework developed in the next section. 

\begin{figure}
\small
\centering
\begin{tabular}{cc|c|c|c|l}
\cline{2-5}
\multicolumn{1}{ c|  }{ } & BS1 & BS2 & BS3 & BS4 & \\
\hhline{|-|-|-|-|-|~|}
\multicolumn{1}{ |c|  }{U11 } & \cellcolor[gray]{0.7}&  & $\times$ & $\times$ & $\leq 
q$\\
\hhline{|-|-|-|-|-|~}
\multicolumn{1}{ |c|  }{U12 } & \cellcolor[gray]{0.7} & $\times$ & $\times$ & & $\leq q$ \\
\hhline{|-|-|-|-|-|~}
\multicolumn{1}{ |c|  }{U21 } & $\times$ & \cellcolor[gray]{0.7} & $\times$ & & $\leq q$ \\
\hhline{|-|-|-|-|-|~}
\multicolumn{1}{ |c|  }{U22 } & $\times$ & \cellcolor[gray]{0.7} &  & $\times$ & $\leq q$ \\
\hhline{|-|-|-|-|-|~}
\multicolumn{1}{ |c|  }{U31 } &  & $\times$ & \cellcolor[gray]{0.7} & $\times$ & $\leq q$ \\
\hhline{|-|-|-|-|-|~}
\multicolumn{1}{ |c|  }{U32 } & $\times$ &  & \cellcolor[gray]{0.7} & $\times$ &$\leq q$ \\
\hhline{|-|-|-|-|-|~}
\multicolumn{1}{ |c|  }{U41 } & $\times$ & $\times$ &  & {\cellcolor[gray]{0.7}} & $\leq q$ \\
\hhline{|-|-|-|-|-|~}
\multicolumn{1}{ |c|  }{U42 } & & $\times$ & $\times$ &  {\cellcolor[gray]{0.7}} & $\leq q$ 	\\
\hhline{|-|-|-|-|-|~}
& \multicolumn{1}{@{}c@{}}{$\underbrace{\hspace*{\dimexpr3\tabcolsep+1\arrayrulewidth}\hphantom{\times} }_{\leq Kq}$} & 
\multicolumn{1}{@{}c@{}}{$\underbrace{\hspace*{\dimexpr3\tabcolsep+1\arrayrulewidth}\hphantom{\times} }_{\leq Kq}$} &
\multicolumn{1}{@{}c@{}}{$\underbrace{\hspace*{\dimexpr3\tabcolsep+1\arrayrulewidth}\hphantom{\times} }_{\leq Kq}$} &
\multicolumn{1}{@{}c@{}}{$\underbrace{\hspace*{\dimexpr3\tabcolsep+1\arrayrulewidth}\hphantom{\times} }_{\leq Kq}$} &
\end{tabular}
\caption{Illustration of the sufficient condition for feasibility of partial IA in a $(4,2,3\times 4)$ cluster where each user can request interference from no more than $q=2$ BSs to be cancelled and each BS can null interference at no more than $Kq=4$ out-of-cell users. }
\label{feasibilityTable}
\end{figure}
\normalsize
\section{Optimization Framework}
\label{section_optframework}

This section focuses on developing an optimization framework capable of leveraging the strength of IA in nulling interference to overcome the limitations imposed by non-convexity of the NUM problem.

In a wireless cellular network, spatial resources can be used in one of three ways: (a) they can be used to serve more users via spatial multiplexing; (b) they can be used to enhance the signal strength (by matched filtering); or (c) they can be used to null  interference (via zero-forcing/IA). NUM algorithms strive to strike the right balance between these  three competing objectives to maximize a certain utility. However, in dense cellular networks, due to the  conflicting nature of these objectives, NUM algorithms may not be able to comprehensively navigate  the entire optimization landscape. The primary motivation behind the proposed approach is to leverage the strength of IA in nulling interference through a pre-optimization step and subsequently using the NUM algorithm to re-balance these priorities to maximize the utility function.

Given a $(G,K,M\times N)$ cluster, we propose a two-stage optimization framework where the first stage focuses on nulling interference from the dominant interferers using IA, followed by a second stage of jointly optimizing the beamformers and the transmit powers to maximize a network utility using the IA  solution as the initial condition. Such a framework is well suited for investigating the benefits of IA in the context of  NUM. The difference in performance with and without the first stage of interference cancellation  sheds light on the value of IA in enhancing the performance of NUM algorithms. For a given network  topology, a significant difference in performance reflects that: (a) IA solutions are valuable from  a NUM perspective; and (b) IA solutions (or close-to-IA solutions) do not organically emerge from  NUM algorithms due to the conflicting uses for spatial resources. 


\begin{figure}[t]
\begin{center}\small
\psfrag{Schedule K users/cell}[cc][cc][0.9]{\raisebox{-3.3mm}{\shifttext{-4mm}{Schedule $K$ users/cell}}}
\psfrag{iterate}[cc][cc][0.9]{iterate}
\psfrag{Identify q dominant}[cc][cc][0.9]{\shifttext{-4mm}{Identify $q$ dominant}}
\psfrag{interfering BS for each user;}[cc][cc][0.9]{\shifttext{-3.5mm}{interferers for each user;}}
\psfrag{ensure feasibility of IA}[cc][cc][0.9]{\shifttext{-0.5mm}{\raisebox{-3mm}{Ensure feasibility of IA}}}
\psfrag{Compute aligned beamformers}[cc][cc][0.9]{\raisebox{-3mm}{\shifttext{-9.8mm}{Compute aligned beamformers}}}
\psfrag{Fix rx. beamformers;}[cc][cc][0.9]{Fix receive beamformers.}
\psfrag{Initialize to aligned beamformers.}[cc][cc][0.87]{\raisebox{-0mm}{\shifttext{6.5mm}{Initialize to aligned beamformers.}}}
\psfrag{Optimize transmit and receive}[cc][cc][0.87]{\shifttext{3.5mm}{Optimize transmit and receive}}
\psfrag{beamformers to maximize}[cc][cc][0.87]{\raisebox{-3mm}{\shifttext{2mm}{beamformers to maximize}}}
\psfrag{the utility function.}[cc][cc][0.87]{\raisebox{-4mm}{\shifttext{6mm}{the utility function.}}}
\psfrag{Optimized power, beamformers}[cc][cc][0.88]{\shifttext{-2.6mm}{Optimized power, beamformers}}
\includegraphics[width=1.8in]{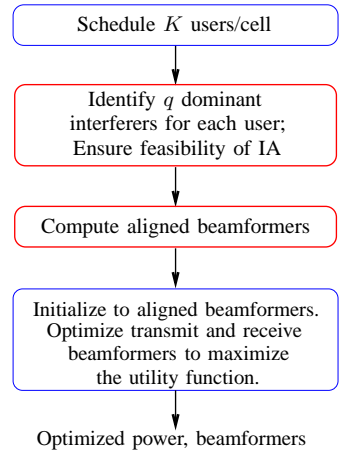}
	\caption{The proposed optimization framework.}
	\label{optflowchart}
      \end{center}
    \end{figure}

\begin{algorithm}[t]\small
\caption{\small Procedure to select dominant interferers such that partial IA is feasible in a $(G,K,M\times N)$ network.}
\label{alg:PIA}
\begin{algorithmic}[1] 
\STATE Fix $q=\lfloor \frac{M+N-1}{K} \rfloor$-1.\\
\STATE For each user identify $q$ dominant interfering BSs based on interference power.\\
\STATE Construct matrix $\bt M$ analogous to that shown in Fig.~\ref{feasibilityTable}. \\
\STATE Identify set of column indices $\mathcal{C}$ in $\bt M$ with greater than $Kq$ chosen entries.\\
\FORALL{ $c \in\mathcal{C}$}
\STATE  Create sorted list $\mathcal{L}_c$  of BS-user pairs from column $c$, sorted  in decreasing order of interference power.\\
\STATE Prune $\mathcal{L}_c$, from the bottom, until no more than $Kq$ pairs remain, to obtain $\hat{\mathcal{L}}_c$.\\
\STATE Use $\hat{\mathcal{L}}_c$ to form column $c$ of $\hat{\bt M}$.
\ENDFOR
\STATE BS-user pairs in $\hat{\bt M}$ ensure feasibility of partial IA.
\end{algorithmic}
\end{algorithm}

Specifically, we evaluate the effectiveness of IA by optimizing a network utility objective of either the sum-rate or the max-min-rate achieved in the network. These two objective functions are chosen specifically to highlight the importance of predetermining the number of users $(K)$ that get served in any given time-frequency slot. While sum-rate maximization allows a user to be assigned no power, thus altering the effective number of users served, such flexibility is not available when maximizing the minimum rate to the set of scheduled users. Since the design of beamformers for IA relies crucially on the number of scheduled users, it is expected that IA has more impact on maximizing the minimum rate than maximizing the sum-rate. This issue is discussed further in Section \ref{section_simulations}. Details of the proposed optimization framework follow.

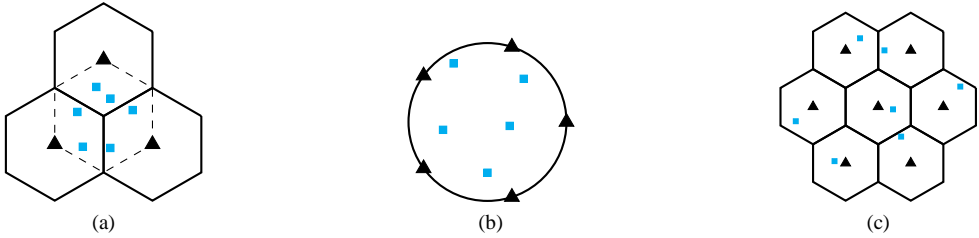
\begin{figure*}
\centering
\subfloat[]{
 \begin{tikzpicture}[scale=0.75, every node/.style={scale=0.75}]
 \coordinate (c1c) at (0,0);
  \foreach \a [count =\i] in {30,90,150,210,270,330}{
    \coordinate (c1\i) at (\a : 1);
  }
  \draw [thick] (c11) -- (c12) -- (c13) -- (c14) -- (c15) -- (c16) -- (c11);
  \coordinate (c2c) at ($(c11)+(90:1)$); 
  \foreach \a [count =\i] in {30,90,150,210,270,330}{
    \coordinate (c2\i) at ($(c2c)+(\a : 1)$);
  }
  \draw [thick] (c21) -- (c22) -- (c23) -- (c24) -- (c25) -- (c26) -- (c21);
  \coordinate (c3c) at ($(c11)+(-30:1)$); 
  \foreach \a [count =\i] in {30,90,150,210,270,330}{
    \coordinate (c3\i) at ($(c3c)+(\a : 1)$);
  }
  \draw [thick] (c31) -- (c32) -- (c33) -- (c34) -- (c35) -- (c36) -- (c31);
  \node[triangle] at (c1c) {};
  \node[triangle] at (c2c) {};
  \node[triangle] at (c3c) {};
  \draw [dashed] (c3c) -- (c26) -- (c2c) -- (c12) -- (c1c) -- (c16) -- (c3c); 
  \node [fill=cyan,draw=cyan, rectangle, inner sep=0, minimum size=1.3mm] at ($(c1c)+(55:0.7)$) (u11) {};
  \node [fill=cyan,draw=cyan, rectangle, inner sep=0, minimum size=1.3mm] at ($(c1c)+(-5:0.5)$) (u12) {};
\node [fill=cyan,draw=cyan, rectangle, inner sep=0, minimum size=1.3mm] at ($(c2c)+(-80:0.7)$) (u21) {};
  \node [fill=cyan,draw=cyan, rectangle, inner sep=0, minimum size=1.3mm] at ($(c2c)+(-105:0.5)$) (u22) {};
\node [fill=cyan,draw=cyan, rectangle, inner sep=0, minimum size=1.3mm] at ($(c3c)+(120:0.7)$) (u31) {};
  \node [fill=cyan,draw=cyan, rectangle, inner sep=0, minimum size=1.3mm] at ($(c3c)+(185:0.75)$) (u32) {};
  \end{tikzpicture}
}\hfil
\subfloat[]{  
  \begin{tikzpicture}[scale=0.75, every node/.style={scale=0.75}]
   \coordinate (c) at (0,0);
   \node[draw=black, circle,thick, inner sep=0, minimum size=2.8cm] (c) {};
   \foreach \i in {1,2,3,4,5}{
   \node[triangle] at (\i * 72 :1.4cm) (BS\i) {};
   }
   \node [fill=cyan,draw=cyan, rectangle, inner sep=0, minimum size=1.3mm] at (50:1) (u1) {};
   \node [fill=cyan,draw=cyan, rectangle, inner sep=0, minimum size=1.3mm] at (120:1.2) (u1) {};
   \node [fill=cyan,draw=cyan, rectangle, inner sep=0, minimum size=1.3mm] at (190:0.8) (u1) {};
   \node [fill=cyan,draw=cyan, rectangle, inner sep=0, minimum size=1.3mm] at (270:0.9) (u1) {};
   \node [fill=cyan,draw=cyan, rectangle, inner sep=0, minimum size=1.3mm] at (350:0.4) (u1) {};
\end{tikzpicture}
}\hfil
\subfloat[]{
  \begin{tikzpicture}[scale=0.5, every node/.style={scale=0.5}]
   \coordinate (c1c) at (0,0);
    \foreach \a [count =\i] in {30,90,150,210,270,330}{
      \coordinate (c1\i) at (\a : 1);
    }
    \draw [thick] (c11) -- (c12) -- (c13) -- (c14) -- (c15) -- (c16) -- (c11);
    \coordinate (c2c) at ($(c11)+(-30:1)$);
    \foreach \a [count =\i] in {30,90,150,210,270,330}{
      \coordinate (c2\i) at ($(c2c)+(\a : 1)$);
    }
    \draw [thick] (c21) -- (c22) -- (c23) -- (c24) -- (c25) -- (c26) -- (c21);
    \coordinate (c3c) at ($(c11)+(90:1)$); 
    \foreach \a [count =\i] in {30,90,150,210,270,330}{
      \coordinate (c3\i) at ($(c3c)+(\a : 1)$);
    }
    \draw [thick] (c31) -- (c32) -- (c33) -- (c34) -- (c35) -- (c36) -- (c31);
    \coordinate (c4c) at ($(c13)+(90:1)$); 
    \foreach \a [count =\i] in {30,90,150,210,270,330}{
      \coordinate (c4\i) at ($(c4c)+(\a : 1)$);
    }
    \draw [thick] (c41) -- (c42) -- (c43) -- (c44) -- (c45) -- (c46) -- (c41);
    \coordinate (c5c) at ($(c13)+(210:1)$); 
    \foreach \a [count =\i] in {30,90,150,210,270,330}{
      \coordinate (c5\i) at ($(c5c)+(\a : 1)$);
    }
    \draw [thick] (c51) -- (c52) -- (c53) -- (c54) -- (c55) -- (c56) -- (c51);
    \coordinate (c6c) at ($(c14)+(-90:1)$); 
    \foreach \a [count =\i] in {30,90,150,210,270,330}{
      \coordinate (c6\i) at ($(c6c)+(\a : 1)$);
    }
    \draw [thick] (c61) -- (c62) -- (c63) -- (c64) -- (c65) -- (c66) -- (c61);
    \coordinate (c7c) at ($(c16)+(-90:1)$);
    \foreach \a [count =\i] in {30,90,150,210,270,330}{
      \coordinate (c7\i) at ($(c7c)+(\a : 1)$);
    }
    \draw [thick] (c71) -- (c72) -- (c73) -- (c74) -- (c75) -- (c76) -- (c71);
    \node[triangle] at (c1c) {};
  \node[triangle] at (c2c) {};
  \node[triangle] at (c3c) {};
  \node[triangle] at (c4c) {};
  \node[triangle] at (c5c) {};
  \node[triangle] at (c6c) {};
  \node[triangle] at (c7c) {};
  \node [fill=cyan,draw=cyan, rectangle, inner sep=0, minimum size=1.3mm] at ($(c1c)+(350:0.4)$) (u1) {};
  \node [fill=cyan,draw=cyan, rectangle, inner sep=0, minimum size=1.3mm] at ($(c2c)+(50:0.70)$) (u1) {};
  \node [fill=cyan,draw=cyan, rectangle, inner sep=0, minimum size=1.3mm] at ($(c3c)+(180:0.7)$) (u1) {};
  \node [fill=cyan,draw=cyan, rectangle, inner sep=0, minimum size=1.3mm] at ($(c4c)+(40:0.5)$) (u1) {};
  \node [fill=cyan,draw=cyan, rectangle, inner sep=0, minimum size=1.3mm] at ($(c5c)+(220:0.6)$) (u1) {};
  \node [fill=cyan,draw=cyan, rectangle, inner sep=0, minimum size=1.3mm] at ($(c6c)+(170:0.3)$) (u1) {};
  \node [fill=cyan,draw=cyan, rectangle, inner sep=0, minimum size=1.3mm] at ($(c7c)+(110:0.75)$) (u1) {};
  \end{tikzpicture}
}
\caption{Network topologies: a three-sector cluster, a 5-cell ring topology and a 7-cell hexagonal layout.}
\label{topologies}
\end{figure*}

\begin{table}
\centering
\caption{Simulation Parameters}
\begin{tabular}{|c|c|c|	}
\hline

\multicolumn{1}{|c|}{\multirow{3}{*}{Network} } &3-Sector 	 &  $(3,K,3\times 4)$ \\
\cline{2-3}
&Ring Topology & $(5,K,5\times 6)$ \\
\cline{2-3}
&Hexagonal Layout & $(7,K,4\times 4)$ \\
\hline
\multicolumn{2}{|c|}{BS-to-BS distance} & \multicolumn{1}{c|}{600m to 1800m}\\  
\hline
\multicolumn{2}{|c|}{Transmit power PSD} & \multicolumn{1}{c|}{-35dBm/Hz} \\
\hline
\multicolumn{2}{|c|}{Thermal noise PSD}  &\multicolumn{1}{c|}{-169dBm/Hz}\\
\hline
\multicolumn{2}{|c|}{Antenna gain} & \multicolumn{1}{c|}{10dBi}\\
\hline
\multicolumn{2}{|c|}{SINR gap} & \multicolumn{1}{c|}{6dB} \\
\hline
\multicolumn{2}{|c|}{Distance dependent pathloss} & \multicolumn{1}{c|}{128.1 +37$\log_{10}(d)$} \\
\hline
\multicolumn{2}{|c|}{Shadowing} &  \multicolumn{1}{c|}{Log-normal, 8dB SD}\\
\hline
\multicolumn{2}{|c|}{Fading} & \multicolumn{1}{c|}{Rayleigh}\\
\hline
\end{tabular}
\label{simparamsTable}
\end{table}

\subsection{Stage I: Partial Interference Alignment}

In the first stage, each user identifies $q$ dominant interferers from whom we attempt to null interference  using IA. The dominant interferers (BSs) are identified based on the strength of the interference caused at the user. Note from Corollary \ref{cor1} that for a given $(G,K,M\times N)$ cluster, the choice of $q$ is closely  dependent on the number of scheduled users; in fact, it is necessary that $q \leq \lfloor  \frac{M+N-1}{K} \rfloor -1$. This suggests that higher the number of scheduled users, fewer  the number of interferers that can be nulled and vice versa. Thus, the number of scheduled users, $K$, emerges as a crucial parameter governing the usefulness of IA.

For a fixed $K$, set $q=\lfloor  \frac{M+N-1}{K} \rfloor -1$. The $q$ dominant interferers are identified by their interference  strength with the transmit and receive beamformers set to certain predetermined values. In our simulations we set all beamformers  to be equal to the all-ones vector.

Once the dominant interferers are identified, we then ensure that the chosen set of BS-user  pairs, denoted as $\mathcal{I}$, conforms to the condition for feasibility of partial IA as stated in  Corollary \ref{cor1}. Constructing a matrix analogous to that shown in Fig.~\ref{feasibilityTable}, it is easy to see that while the rows of this matrix have no more than $q$ chosen entries by  construction, the columns may have  more than $Kq$ chosen entries. To eliminate such cases, if any column has more than $Kq$ chosen  cells, we sort the chosen  cells of this column in the descending order of their interference strengths and prune this sorted  list, from the bottom,  until no more than $Kq$ cells are left. The set of BS-user pairs that results at the end of this  process (denoted as $\tilde{\mathcal{I}}$), satisfies the conditions imposed by  Corollary \ref{cor1} thus ensuring the feasibility of partial IA. As a result of the pruning, not all users have interference from all their $q$ dominant  interferers  nulled; but on average many if not most of them do. Note that for the case  $q=G-1$, no such pruning is necessary. An outline of the above procedure is given in Algorithm \ref{alg:PIA}.

Once $\tilde{\mathcal{I}}$ is obtained, aligned beamformers satisfying the conditions for partial IA can be designed using any algorithm developed for IA such as interference leakage minimization \cite{gomadam,zhuangfeasibility,guillaudgesbert}, iterative matrix norm minimization \cite{gokulTSP}, etc. 

\subsection{Stage II: Utility Maximization}

This stage focuses on maximizing a given network utility function using the aligned beamformers obtained in the previous stage as the initialization. As stated before, this paper focuses on maximizing either the sum-rate or the minimum rate for the scheduled users subject to per-BS power constraints. The proposed optimization framework is outlined in Fig.~\ref{optflowchart}. A brief description of the optimization problems that need to be solved for utility maximization follows.

\subsubsection{Sum-rate maximization}
Maximizing the sum-rate requires solving the following optimization problem. 
\begin{align}
\maximize_{\bt v_{gk},\ \bt u_{gk}} & \quad \sum_{g,k}  \log  \left (1+\tfrac{| \bt u^H_{gk} \bt H_{(g,gk)} \bt v_{gk}|^2}{\sigma_{n}^2+\nu_{gk}^2+\sum \limits_{(i,j)\neq (g,k)}|\bt u^H_{gk} \bt H_{(i,gk)} \bt v_{ij}|^2} \right ) \nonumber \\
\text{subject to} & \quad \sum_{k=1}^K \bt |\bt v_{gk}|^2\leq P_{max}, \quad \forall g, \nonumber \\
& \quad |\bt u_{gk}|^2=1, \quad \forall (g,k),
\label{sumratemax}
\end{align}
where $\sigma_{n}^2$ represents the variance of additive noise and $\nu_{gk}^2$ is the variance of out-of-cluster interference (for isolated clusters, this term is set to zero). No convex reformulations of the above problem are known and hence one can at best hope to obtain a locally optimal solution. A locally optimal solution can be obtained through a computationally efficient algorithm, proposed in \cite{christensen,wmmse}, known as the weighted minimum mean-squared error (WMMSE) algorithm. For further details on this algorithm refer to \cite{christensen} and \cite{wmmse}. The algorithm is initialized to the aligned transmit and receive beamformers obtained from the pre-optimization step.

\subsubsection{Max-min fairness}
In order to maximize the minimum user rate achieved by the set of scheduled users in the given cluster, we solve the following optimization problem:
\begin{align}
\maximize_{\bt v_{gk},\ \bt u_{gk}} & \quad t \nonumber \\
\text{subject to} & \quad \tfrac{| \bt u^H_{gk} \bt H_{(g,gk)} \bt v_{gk}|^2}{\sigma_n^2+\nu_{gk}^2+\sum \limits_{(i,j)\neq (g,k)}|\bt u^H_{gk} \bt H_{(i,gk)} \bt v_{ij}|^2} \geq t, \quad \forall (g,k),\nonumber \\
& \quad \sum_{k=1}^K \bt |\bt v_{gk}|^2\leq P_{max}, \quad \forall g, \nonumber \\
& \quad |\bt u_{gk}|^2=1, \quad \forall (g,k),
\label{minmaxrate}
\end{align}
where $\bt v_{gk}$, $\bt u_{gk}$ are the variables for optimization, $P_{max}$ is the maximum transmit power permitted at any BS and $\sigma_{n}^2$ and $\nu^2_{gk}$ are as defined earlier. This problem is non-convex in its current form and no convex reformulation is known except when the users have a single  antenna. Several techniques for finding a local optimum of this problem have been proposed  \cite{minmaxyichao,minmaxraza,minmaxCWT}. We solve (\ref{minmaxrate}) by   alternately optimizing the transmit and receive beamformers, leveraging the  convex reformulation that emerges when users have a single antenna \cite{wiesel}.  Fixing the receive beamformers to be the aligned beamformers obtained from the first stage, we use  a bisection search over $t$ to find a maximal min-rate as proposed in \cite{wiesel}. Fixing  the transmit beamformers to those obtained at the end of this bisection search, the optimal receive  beamformers are given by the MMSE beamformers. Once the receive beamformers are updated, we proceed to re-optimize the transmit beamformers and this procedure is repeated for a fixed number of iterations.

\section{Simulation Results}
\label{section_simulations}
\subsection{Isolated Clusters}

The value of IA is best illustrated in a dense cluster of isolated BSs where interference mitigation plays  an increasingly important role as the distance between BSs decreases. Towards this end, we consider  three network topologies with increasing cluster sizes to test the proposed framework. As shown in Fig.~\ref{topologies}, the first  network is a 3-sector cluster, the second consists of 5 BSs spread out on a ring and the  third is a 7-cell hexagonal cluster. Same pathloss, shadowing and fading assumptions are made for  all three networks. Users are assumed to be uniformly distributed in each cell, and are served by  one data stream each. Table \ref{simparamsTable} lists the antenna configuration for each of the networks, along with other parameter settings.

For each network, the number of scheduled users per cell, $K$, is varied from $\left \lfloor\frac{M+N-1}{G} \right \rfloor$ to  $N$. Note that as $K$ increases, the number of dominant BSs that can be cancelled in the first  stage decreases. When $K>\frac{M+N-1}{2}$, no dominant interferers can be nulled and the  beamformers are chosen to only cancel intra-cell interference.

For a given set of scheduled users, the proposed optimization framework is used to maximize either the minimum user rate or the sum-rate. For each user, interference from at most  $q= \lfloor \frac{M+N-1}{K} \rfloor -1$ interferers is nulled using the  interference leakage minimization algorithm \cite{gomadam}. Using these aligned beamformers as initialization, the optimization problem presented in (\ref{sumratemax}) or (\ref{minmaxrate}) is solved depending on the choice of the utility function. The algorithm in \cite{wmmse} is used to solve (\ref{sumratemax}) and is run until convergence. To solve (\ref{minmaxrate}), transmit and receive beamformers are alternately optimized for a fixed number of iterations. The convex optimization problem arising from (\ref{minmaxrate}) for a fixed set of receive  beamformers is solved using CVX, a package for specifying and solving convex programs \cite{cvx1,cvx2}. The performance of the proposed framework is compared to the setup where the  first stage is omitted, i.e., the dominant interferers are not nulled using IA  (marked as `no IA'). The results of the optimization are averaged over 100 user locations.

\begin{figure}[t]
\begin{center}
\psfrag{1}[cc][cc][0.75]{$1$}
\psfrag{4}[cc][lc][0.75]{\shifttext{-2mm}{$4$}}
\psfrag{6}[cc][lc][0.75]{\shifttext{-2mm}{$6$}}
\psfrag{8}[cc][lc][0.75]{\shifttext{-2mm}{$8$}}
\psfrag{10}[cc][lc][0.75]{$10$}
\psfrag{12}[cc][lc][0.75]{$12$}
\psfrag{14}[cc][lc][0.75]{$14$}
\psfrag{16}[cc][lc][0.75]{$16$}
\psfrag{18}[cc][lc][0.75]{$18$}
\psfrag{x1}[cc][cc][0.75]{$600$}
\psfrag{x2}[cc][cc][0.75]{$900$}
\psfrag{x3}[cc][cc][0.75]{$1200$}
\psfrag{x4}[cc][cc][0.75]{$1500$}
\psfrag{x5}[cc][cc][0.75]{$1800$}
\psfrag{xlabel}[tc][cc][0.7]{distance between BSs (in meters)}
\psfrag{ylabel}[Bc][tc][0.7][0]{Average cell throughput (b/s/Hz)}
\psfrag{ILM}[Bl][Bl][0.72]{ILM}
\psfrag{K=2, no IA}[Bl][Bl][0.65]{$K=2$, no IA}
\psfrag{K=3, no IA}[Bl][Bl][0.65]{$K=3$, no IA}
\psfrag{K=4, no IA}[Bl][Bl][0.65]{$K=4$, no IA}
\psfrag{K=2, with IA, q=2}[Bl][Bl][0.65]{$K=2$, with IA, $q=2$}
\psfrag{K=3, with IA, q=1}[Bl][Bl][0.65]{$K=3$, with IA, $q=1$}
\psfrag{K=4, with IA, q=0}[Bl][Bl][0.65]{$K=4$, with IA, $q=0$}
\includegraphics[width=3.4in]{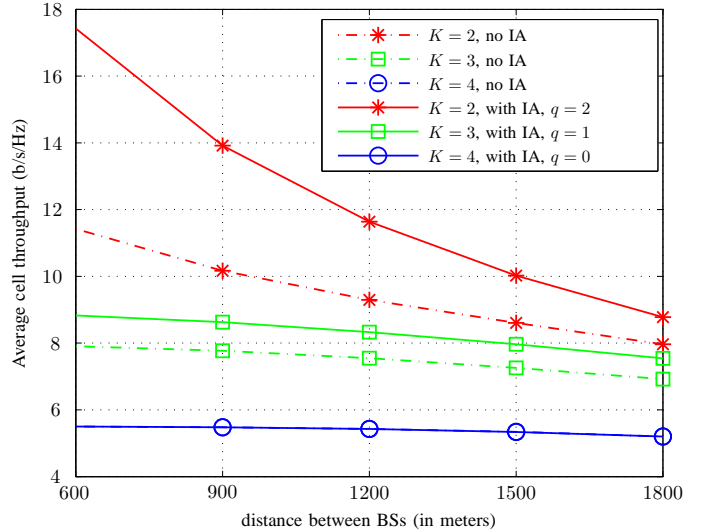}
	\caption{Average per-cell throughput in a $(3,K,3\times 4)$ network forming a 3-sector cluster when maximizing minimum user rate under per-BS power constraints. Cell-throughput is defined as $K$ times the minimum user rate.}
	\label{C3Riso}
      \end{center}
    \end{figure}


\subsubsection{Maximizing the minimum rate}

\begin{figure*}[t]
\centering
\subfloat[]{
\psfrag{7}[cc][lc][0.65]{\shifttext{-2mm}{$7$}}
\psfrag{8}[cc][lc][0.65]{\shifttext{-2mm}{$8$}}
\psfrag{9}[cc][lc][0.65]{\shifttext{-2mm}{$9$}}
\psfrag{10}[cc][lc][0.65]{$10$}
\psfrag{11}[cc][lc][0.65]{$11$}
\psfrag{12}[cc][lc][0.65]{$12$}
\psfrag{13}[cc][lc][0.65]{$13$}
\psfrag{14}[cc][lc][0.65]{$14$}
\psfrag{15}[cc][lc][0.65]{$15$}
\psfrag{16}[cc][lc][0.65]{$16$}
\psfrag{17}[cc][lc][0.65]{$17$}
\psfrag{x1}[cc][cc][0.65]{$600$}
\psfrag{x2}[cc][cc][0.65]{$900$}
\psfrag{x3}[cc][cc][0.65]{$1200$}
\psfrag{x4}[cc][cc][0.65]{$1500$}
\psfrag{x5}[cc][cc][0.65]{$1800$}
\psfrag{xlabel}[tc][cc][0.75]{distance between BSs (in meters)}
\psfrag{ylabel}[Bc][tc][0.75][0]{Average cell throughput (b/s/Hz)}
\psfrag{ILM}[Bl][Bl][0.72]{ILM}
\psfrag{K=2, no IA}[Bl][Bl][0.56]{$K=2$, no IA}
\psfrag{K=3, no IA}[Bl][Bl][0.56]{$K=3$, no IA}
\psfrag{K=4, no IA}[Bl][Bl][0.56]{$K=4$, no IA}
\psfrag{K=5, no IA}[Bl][Bl][0.56]{$K=5$, no IA}
\psfrag{K=6, no IA}[Bl][Bl][0.56]{$K=6$, no IA}
\psfrag{K=2, with IA, q=4}[Bl][Bl][0.56]{$K=2$, with IA, $q=4$}
\psfrag{K=3, with IA, q=2}[Bl][Bl][0.56]{$K=3$, with IA, $q=2$}
\psfrag{K=4, with IA, q=1}[Bl][Bl][0.56]{$K=4$, with IA, $q=1$}
\psfrag{K=5, with IA, q=1}[Bl][Bl][0.56]{$K=5$, with IA, $q=1$}
\psfrag{K=6, with IA, q=0}[Bl][Bl][0.56]{$K=6$, with IA, $q=0$}
\includegraphics[width=2.9in]{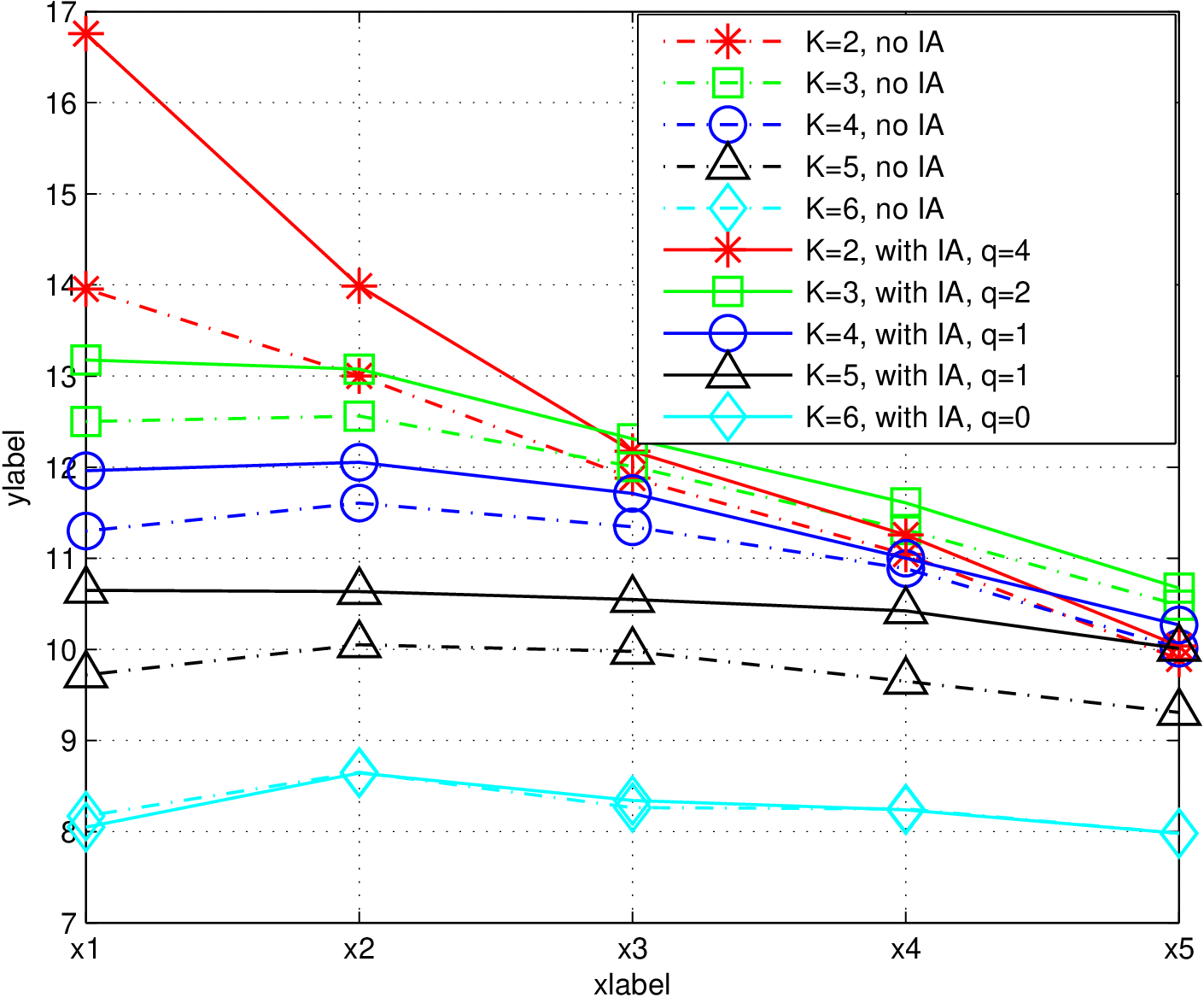}
	\label{C5Riso}
}\hfil
\subfloat[]{
\psfrag{4}[cc][lc][0.65]{\shifttext{-2mm}{$4$}}
\psfrag{6}[cc][lc][0.65]{\shifttext{-2mm}{$6$}}
\psfrag{7}[cc][lc][0.65]{\shifttext{-2mm}{$7$}}
\psfrag{8}[cc][lc][0.65]{\shifttext{-2mm}{$8$}}
\psfrag{9}[cc][lc][0.65]{\shifttext{-2mm}{$9$}}
\psfrag{10}[cc][lc][0.65]{$10$}
\psfrag{6.5}[cc][cc][0.65]{}
\psfrag{7.5}[cc][cc][0.65]{}
\psfrag{8.5}[cc][cc][0.65]{}
\psfrag{9.5}[cc][cc][0.65]{}
\psfrag{12}[cc][lc][0.65]{$12$}
\psfrag{14}[cc][lc][0.65]{$14$}
\psfrag{16}[cc][lc][0.65]{$16$}
\psfrag{18}[cc][lc][0.65]{$18$}
\psfrag{x1}[cc][cc][0.65]{$600$}
\psfrag{x2}[cc][cc][0.65]{$900$}
\psfrag{x3}[cc][cc][0.65]{$1200$}
\psfrag{x4}[cc][cc][0.65]{$1500$}
\psfrag{x5}[cc][cc][0.65]{$1800$}
\psfrag{xlabel}[tc][cc][0.85]{distance between BSs (in meters)}
\psfrag{ylabel}[Bc][tc][0.85][0]{Average cell throughput (b/s/Hz)}
\psfrag{ILM}[Bl][Bl][0.72]{ILM}
\psfrag{K=1, no IA}[Bl][Bl][0.58]{$K=1$, no IA}
\psfrag{K=2, no IA}[Bl][Bl][0.58]{$K=2$, no IA}
\psfrag{K=3, no IA}[Bl][Bl][0.58]{$K=3$, no IA}
\psfrag{K=4, no IA}[Bl][Bl][0.58]{$K=4$, no IA}
\psfrag{K=1, with IA, q=6}[Bl][Bl][0.58]{$K=1$, with IA, $q=6$}
\psfrag{K=2, with IA, q=2}[Bl][Bl][0.58]{$K=2$, with IA, $q=2$}
\psfrag{K=3, with IA, q=1}[Bl][Bl][0.58]{$K=3$, with IA, $q=1$}
\psfrag{K=4, with IA, q=0}[Bl][Bl][0.58]{$K=4$, with IA, $q=0$}
\includegraphics[width=3.08in]{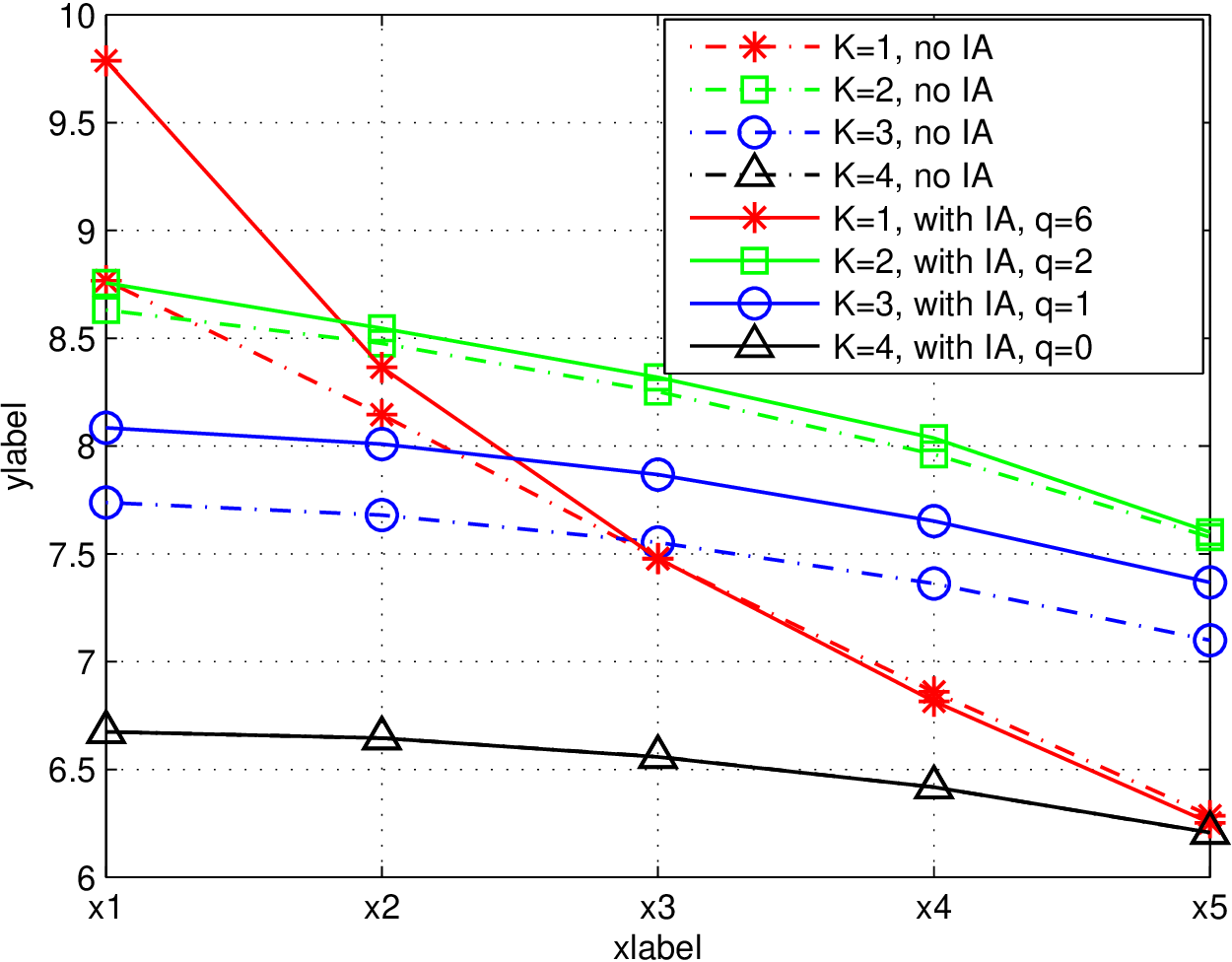}
	\label{C7Riso}
}
      \caption{Average per-cell throughput in (a) $(5,K,5\times 6)$ ring topology and (b) $(7,K,4\times 4)$ hexagonal layout when maximizing minimum user rate under per-BS power constraints. Cell-throughput is defined as $K$ times the minimum user rate.}
    \end{figure*}

\begin{figure*}[t]
\centering
\subfloat[]{
\psfrag{1}[cc][cc][0.75]{$1$}
\psfrag{4}[cc][lc][0.75]{\shifttext{-2mm}{$4$}}
\psfrag{6}[cc][lc][0.75]{\shifttext{-2mm}{$6$}}
\psfrag{8}[cc][lc][0.75]{\shifttext{-2mm}{$8$}}
\psfrag{9}[rc][lc][0.75]{$9$}
\psfrag{10}[cc][lc][0.75]{$10$}
\psfrag{12}[cc][lc][0.75]{$12$}
\psfrag{14}[cc][lc][0.75]{$14$}
\psfrag{16}[cc][lc][0.75]{$16$}
\psfrag{18}[cc][lc][0.75]{$18$}
\psfrag{11}[cc][lc][0.75]{$11$}
\psfrag{13}[cc][lc][0.75]{$13$}
\psfrag{15}[cc][lc][0.75]{$15$}
\psfrag{17}[cc][lc][0.75]{$17$}
\psfrag{x1}[cc][cc][0.75]{$600$}
\psfrag{x2}[cc][cc][0.75]{$900$}
\psfrag{x3}[cc][cc][0.75]{$1200$}
\psfrag{x4}[cc][cc][0.75]{$1500$}
\psfrag{x5}[cc][cc][0.75]{$1800$}
\psfrag{xlabel}[tc][cc][0.85]{\raisebox{-0.0cm}{distance between BSs (in meters)}}
\psfrag{ylabel}[Bc][tc][0.85][0]{Average cell throughput (b/s/Hz)}
\psfrag{ILM}[Bl][Bl][0.72]{ILM}
\psfrag{K=2, no IA}[Bl][Bl][0.65]{$K=2$, no IA}
\psfrag{K=3, no IA}[Bl][Bl][0.65]{$K=3$, no IA}
\psfrag{K=4, no IA}[Bl][Bl][0.65]{$K=4$, no IA}
\psfrag{K=2, with IA, q=2}[Bl][Bl][0.65]{$K=2$, with IA, $q=2$}
\psfrag{K=3, with IA, q=1}[Bl][Bl][0.65]{$K=3$, with IA, $q=1$}
\psfrag{K=4, with IA, q=0}[Bl][Bl][0.65]{$K=4$, with IA, $q=0$}
\includegraphics[width=3in]{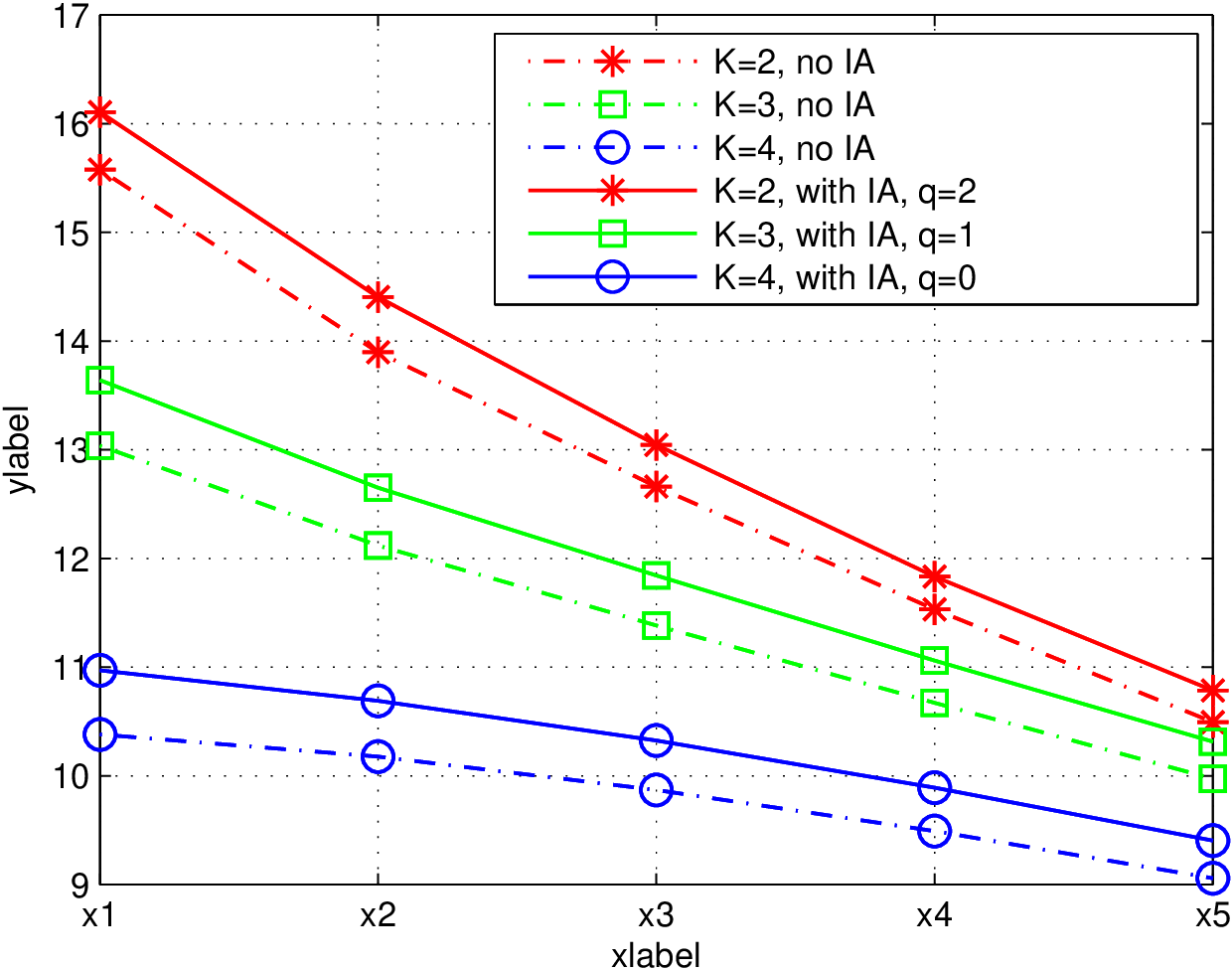}
	\label{C3Risowmmse}
}\hfil
\subfloat[]{
\psfrag{4}[cc][lc][0.75]{\shifttext{-2mm}{$4$}}
\psfrag{6}[cc][lc][0.75]{\shifttext{-2mm}{$6$}}
\psfrag{7}[cc][lc][0.75]{\shifttext{-2mm}{$7$}}
\psfrag{8}[cc][lc][0.75]{\shifttext{-2mm}{$8$}}
\psfrag{9}[cc][lc][0.75]{\shifttext{-2mm}{$9$}}
\psfrag{10}[cc][lc][0.75]{$10$}
\psfrag{11}[cc][lc][0.75]{$11$}
\psfrag{6.5}[cc][cc][0.75]{}
\psfrag{7.5}[cc][cc][0.75]{}
\psfrag{8.5}[cc][cc][0.75]{}
\psfrag{9.5}[cc][cc][0.75]{}
\psfrag{10.5}[cc][cc][0.75]{}
\psfrag{11.5}[cc][cc][0.75]{}
\psfrag{12}[cc][lc][0.75]{$12$}
\psfrag{14}[cc][lc][0.75]{$14$}
\psfrag{16}[cc][lc][0.75]{$16$}
\psfrag{18}[cc][lc][0.75]{$18$}
\psfrag{x1}[cc][cc][0.75]{$600$}
\psfrag{x2}[cc][cc][0.75]{$900$}
\psfrag{x3}[cc][cc][0.75]{$1200$}
\psfrag{x4}[cc][cc][0.75]{$1500$}
\psfrag{x5}[cc][cc][0.75]{$1800$}
\psfrag{xlabel}[tc][cc][0.85]{\raisebox{0.0cm}{distance between BSs (in meters)}}
\psfrag{ylabel}[Bc][tc][0.85][0]{\raisebox{-0.27cm}{Average cell throughput (b/s/Hz)}}
\psfrag{ILM}[Bl][Bl][0.72]{ILM}
\psfrag{K=1, no IA}[Bl][Bl][0.55]{$K=1$, no IA}
\psfrag{K=2, no IA}[Bl][Bl][0.55]{$K=2$, no IA}
\psfrag{K=3, no IA}[Bl][Bl][0.55]{$K=3$, no IA}
\psfrag{K=4, no IA}[Bl][Bl][0.55]{$K=4$, no IA}
\psfrag{K=1, with IA, q=6}[Bl][Bl][0.55]{$K=1$, with IA, $q=6$}
\psfrag{K=2, with IA, q=2}[Bl][Bl][0.55]{$K=2$, with IA, $q=2$}
\psfrag{K=3, with IA, q=1}[Bl][Bl][0.55]{$K=3$, with IA, $q=1$}
\psfrag{K=4, with IA, q=0}[Bl][Bl][0.55]{$K=4$, with IA, $q=0$}
\includegraphics[width=3.06in]{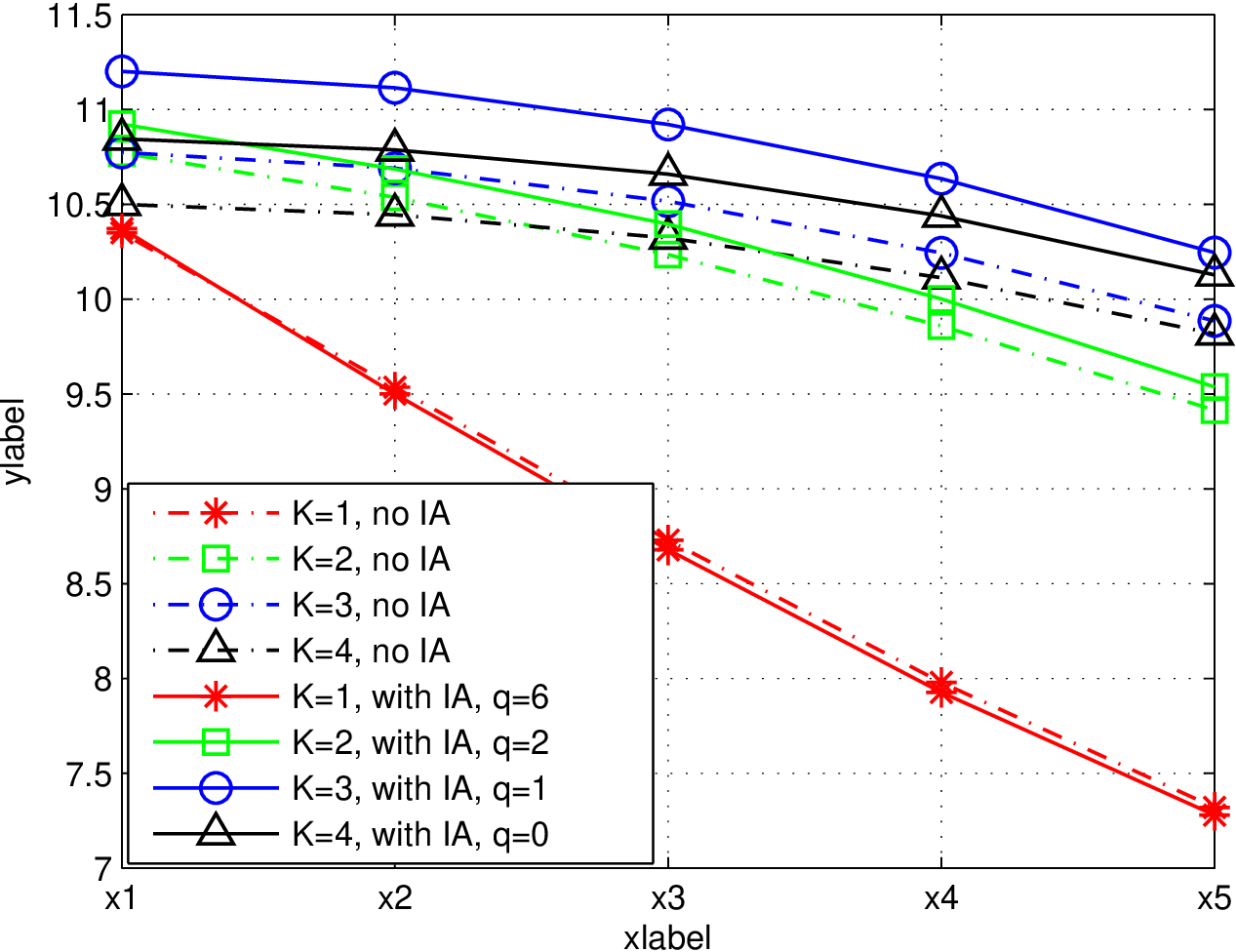}
	\label{C7Risowmmse}
      }
      \caption{Average per-cell throughput in (a) $(3,K,3\times 4)$ three-sector cluster and (b) $(7,K,4\times 4)$ hexagonal layout when maximizing the sum-rate. The algorithm in \cite{wmmse} is used to maximize the sum-rate under per-BS power constraints.}
    \end{figure*}

Figs. \ref{C3Riso}, \ref{C5Riso} and \ref{C7Riso} plot the results of maximizing the minimum rate for each of  the three networks as a function of BS-to-BS distance and the number of scheduled users. Average cell throughput---measured as the max-min rate times the number of scheduled users $(K)$---is used as the performance metric for comparison. It is seen that IA solutions  provide an altered interference landscape that is otherwise non-trivial to find, and this altered landscape enhances the  performance of subsequent NUM algorithms. Focusing on Fig.~\ref{C3Riso}, it is clear that IA has a significant impact on  optimization, especially when BSs are closely spaced. The gain of IA depends on the number of users scheduled. In  particular, when 2~users/cell are scheduled, it is possible to achieve 1~DoF/user as interference can be completely  nulled in the network ($q=G-1=2$). In this case, IA provides 4-6~b/s/Hz improvement at small BS-to-BS distances. When  3~users/cell are scheduled, IA can cancel interference from up to one interferer for each user. Such IA  solutions are seen to enhance the average cell throughput by about 1~b/s/Hz. However, when 4~users/cell  are scheduled, only intra-cell interference can be nulled, and IA has no impact on the  optimization. Note also that because it is possible to completely null inter-cell interference only when  $K=2$ (or equivalently, $q=2$), this is the only scenario where throughput does not saturate as the BS-to-BS distance decreases. Finally, we comment that for a broad range of BS-to-BS distances, scheduling 2~users/cell appears to be optimal.

A similar set of observations can also be made in Fig.~\ref{C5Riso}. In particular, IA provides about 1~b/s/Hz gain when $K\leq 5$ and over a good range of BS-to-BS distances. However, unlike the 3-sector  network, nulling interference from all interferers (i.e., $q=4$, $K=2$) is not necessarily the best strategy, except at  very  small BS-to-BS distances. At larger distances it appears that nulling interference from the two dominant interferers  suffices ($q=2$, $K=3$).

Finally, Fig.~\ref{C7Riso} considers the 7-cell network---the only network, among the three  considered here, where not all cells are equivalent and pruning the list of dominant interferers  plays an important role in ensuring feasibility of partial IA. As expected, it can be seen that with  increasing cluster size, scheduling $K=\lfloor \frac{M+N-1}{G} \rfloor$ users (in this case, $K=1$, $q=6$), is a good strategy only at small BS-to-BS distances. In fact IA does not provide consistent rate  gain across all the cases. But the simulation does provide insight on the optimal number of users to schedule. It  appears that the number of  scheduled users should be such that nulling interference from one or two of the dominant  interferers for each user is feasible.

Surprisingly, in all three networks, scheduling as many users as  there are antennas does not appear to be the right choice even at large BS-to-BS distances.  Aggressive spatial multiplexing seems to severely limit the use of spatial resources to enhance  signal strength or to null interference.

\begin{figure*}[t]
\centering
\subfloat[]{
\psfrag{1}[cc][cc][0.75]{$1$}
\psfrag{0.4}[cc][cc][0.75]{\shifttext{-0mm}{$0.4$}}
\psfrag{0.6}[cc][cc][0.75]{\shifttext{-0mm}{$0.6$}}
\psfrag{0.8}[cc][cc][0.75]{\shifttext{-0mm}{$0.8$}}
\psfrag{0.2}[cc][cc][0.75]{\shifttext{-0mm}{$0.2$}}
\psfrag{0}[cc][cc][0.75]{$0$}
\psfrag{-10}[cc][cc][0.75]{$-10$}
\psfrag{-5}[cc][cc][0.75]{$-5$}
\psfrag{5}[cc][cc][0.75]{$5$}
\psfrag{10}[cc][cc][0.75]{$10$}
\psfrag{15}[cc][cc][0.75]{$15$}
\psfrag{20}[cc][cc][0.75]{$20$}
\psfrag{xtitle}[cc][cc][0.85]{Transmit power in dBm}
\psfrag{ytitle}[Bc][tc][0.85][0]{}
\psfrag{title}[Bc][tc][0.85][0]{}
\psfrag{with IA}[Bl][Bl][0.74]{with IA}
\psfrag{without IA}[Bl][Bl][0.74]{without IA}
\psfrag{K=2, no IA}[Bl][Bl][0.65]{$K=2$, no IA}
\psfrag{K=3, no IA}[Bl][Bl][0.65]{$K=3$, no IA}
\psfrag{K=4, no IA}[Bl][Bl][0.65]{$K=4$, no IA}
\psfrag{K=2, with IA, q=2}[Bl][Bl][0.65]{$K=2$, with IA, $q=2$}
\psfrag{K=3, with IA, q=1}[Bl][Bl][0.65]{$K=3$, with IA, $q=1$}
\psfrag{K=4, with IA, q=0}[Bl][Bl][0.65]{$K=4$, with IA, $q=0$}
\includegraphics[width=3.06in]{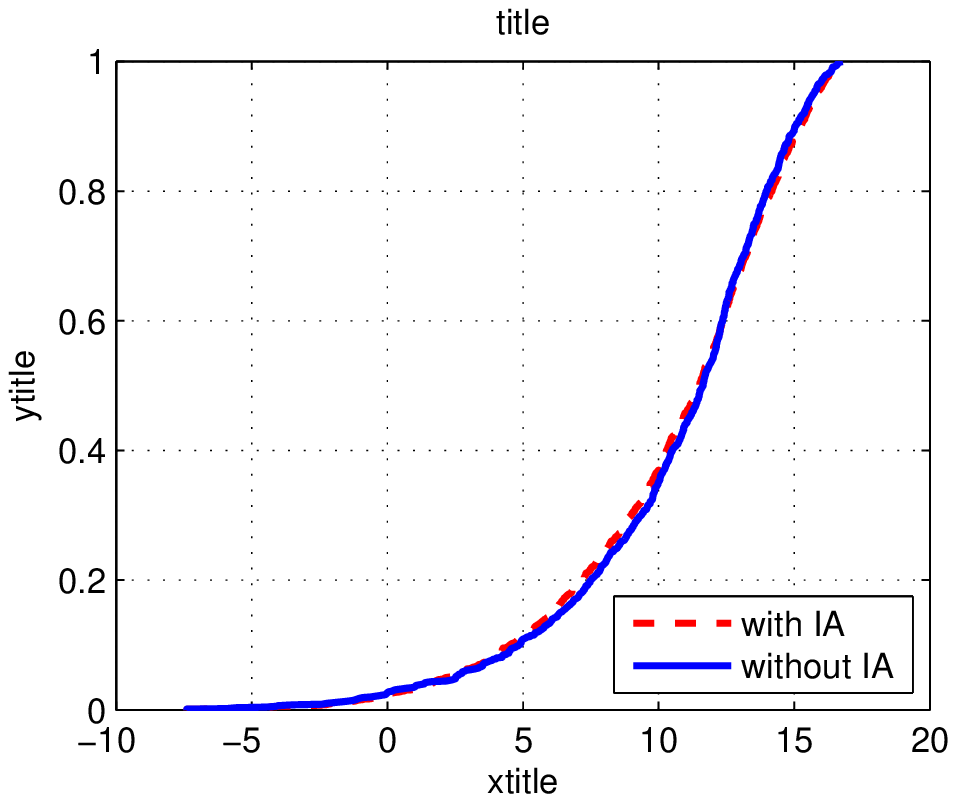}
	\label{powerCDF}
}\hfil
\subfloat[]{
\psfrag{0.2}[cc][cc][0.75]{\shifttext{-0mm}{$0.2$}}
\psfrag{0.4}[cc][cc][0.75]{\shifttext{-0mm}{$0.4$}}
\psfrag{0.6}[cc][cc][0.75]{\shifttext{-0mm}{$0.6$}}
\psfrag{1}[cc][cc][0.75]{\shifttext{-0mm}{$1$}}
\psfrag{0.8}[cc][cc][0.75]{\shifttext{-0mm}{$0.8$}}
\psfrag{9}[cc][lc][0.75]{\shifttext{-2mm}{$9$}}
\psfrag{10}[cc][lc][0.75]{$10$}
\psfrag{11}[cc][lc][0.75]{$15$}
\psfrag{6.5}[cc][cc][0.75]{}
\psfrag{7.5}[cc][cc][0.75]{}
\psfrag{8.5}[cc][cc][0.75]{}
\psfrag{9.5}[cc][cc][0.75]{}
\psfrag{60}[cc][cc][0.75]{$60$}
\psfrag{40}[cc][cc][0.75]{$40$}
\psfrag{20}[cc][cc][0.75]{$20$}
\psfrag{0}[cc][cc][0.75]{$0$}
\psfrag{-20}[cc][cc][0.75]{$-20$}
\psfrag{0}[cc][cc][0.75]{$0$}
\psfrag{xtitle}[cc][cc][0.85]{\shifttext{-10mm}{SINR in dB}}
\psfrag{ytitle}[Bc][tc][0.85][0]{}
\psfrag{title}[Bc][tc][0.85][0]{}
\psfrag{with IA}[Bl][Bl][0.74]{with IA}
\psfrag{without IA}[Bl][Bl][0.74]{without IA}
\psfrag{K=1, no IA}[Bl][Bl][0.55]{$K=1$, no IA}
\psfrag{K=2, no IA}[Bl][Bl][0.55]{$K=2$, no IA}
\psfrag{K=3, no IA}[Bl][Bl][0.55]{$K=3$, no IA}
\psfrag{K=4, no IA}[Bl][Bl][0.55]{$K=4$, no IA}
\psfrag{K=1, with IA, q=6}[Bl][Bl][0.55]{$K=1$, with IA, $q=6$}
\psfrag{K=2, with IA, q=2}[Bl][Bl][0.55]{$K=2$, with IA, $q=2$}
\psfrag{K=3, with IA, q=1}[Bl][Bl][0.55]{$K=3$, with IA, $q=1$}
\psfrag{K=4, with IA, q=0}[Bl][Bl][0.55]{$K=4$, with IA, $q=0$}
\includegraphics[width=3.06in]{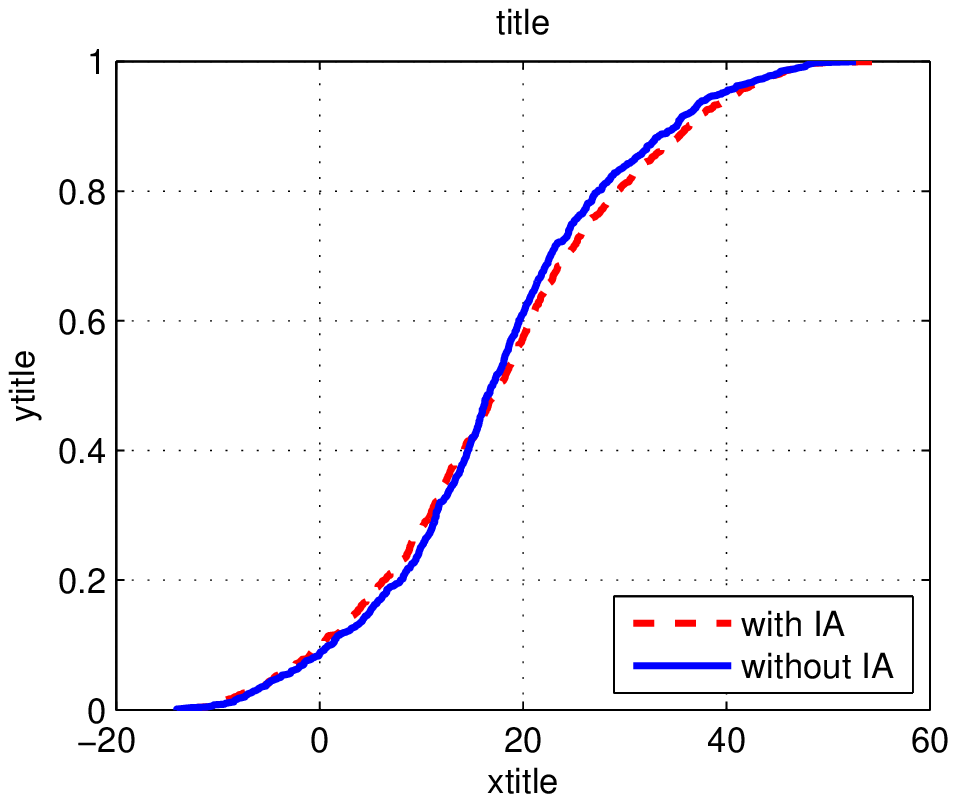}
	\label{sinrCDF}
      }
      \caption{Post-optimization CDFs of transmit powers (per user) and SINRs in a  $(3,3,3\times 4)$ network.}
      \label{CDFplot}
    \end{figure*}

\subsubsection{Maximizing the sum-rate}

The observations made above for maximizing the minimum rate are not necessarily applicable for maximizing the sum-rate. Figs. \ref{C3Risowmmse} and \ref{C7Risowmmse} plot the performance of the proposed framework with and without IA for the 3-sector and 7-cell topologies. It is clear that IA makes only a marginal difference to the overall throughput. The difference between the two utility functions can be explained by noting that when maximizing the sum-rate, unlike max-min fairness, even though $K$ users are scheduled in a time-frequency slot only a subset of these users get prioritized and increase in their throughput comes at the cost of the other scheduled users who are either allocated very little transmit power or face significant interference. This can be seen in Figs. \ref{powerCDF} and \ref{sinrCDF} where the cumulative distribution functions (CDFs) of transmit power and SINRs are plotted after the two-step optimization for the 3-sector $(3,3,3\times 4)$ network. In such a network there are a total of 9 users scheduled at each instance, and IA aims to null interference from 1 BS for each user. It can be seen that about $10\%$ of the users ultimately end up with an SINR less than 0 dB (equivalent to 1 user in every scheduling instance), while another $10\%$ of the users achieve an SINR exceeding 35 dB. Wide disparity in transmit power allocation can also be observed, with over $10\%$ of users receiving more than 15 dBm of transmit power out of a maximum of 16.9 dBm (per cell, per tone). This flexibility in prioritizing users and assigning resources undermines the value of the aligned beamformers that have been designed under the assumption that all $K$ users in a cell are equally important, leading to only marginal gains due to IA when maximizing the sum-rate. This also suggests that the post-optimization interference landscape, with changes in transmit power allocation and number of users with active transmissions (i.e., SINRs above a certain threshold), is so different that the initial assumptions on the dominant interferers are rendered irrelevant.


\subsection{Non-isolated Clusters}
It is also important to test the effectiveness of IA in an environment where the given cluster of cooperating BSs is surrounded by other non-cooperating BSs, which produce out-of-cluster interference. Towards this end, we simulate a 49-cell network forming a hexagonal topology with the central 7 cells forming a cluster similar to that shown in Fig.~\ref{topologies}. Thus, users in these 7 cells see out-of-cluster interference from 35 other BSs that surround them. Applying the proposed framework in such an environment while treating out-of-cluster interference as noise, it is seen from Fig.~\ref{C7RnonisominmaxA} that (a) density has little impact on the overall throughput and (b) aligned beamformers carry little significance. While not presented here, a similar set of results are obtained when maximizing the sum-rate as well. It is clear from the spectral efficiencies achieved that such environments are significantly limited by out-of-cluster interference. Nulling interference from a few dominant interferers while ignoring signal strength does not impact the final outcome of the optimization. These results suggest that when investigating beamformer design in practical cellular environments, focusing exclusively on the design of aligned beamformers does not warrant sufficient importance and that in such circumstances, more attention must be paid to the performance of NUM algorithms under various practical constraints such as CSI acquisition, etc.

\section{Conclusion}

\begin{figure}[t]
\begin{center}
\psfrag{y1}[cc][cc][0.75]{\shifttext{-3mm}{$2$}}
\psfrag{y2}[cc][cc][0.75]{\shifttext{-3mm}{$2.5$}}
\psfrag{y3}[cc][cc][0.75]{\shifttext{-3mm}{$3$}}
\psfrag{y4}[cc][cc][0.75]{\shifttext{-3mm}{$3.5$}}
\psfrag{y5}[cc][cc][0.75]{\shifttext{-3mm}{$4$}}
\psfrag{0y}[cc][cc][0.75]{\shifttext{-3mm}{$0$}}
\psfrag{11}[cc][lc][0.75]{$11$}
\psfrag{6.5}[cc][cc][0.75]{}
\psfrag{7.5}[cc][cc][0.75]{}
\psfrag{8.5}[cc][cc][0.75]{}
\psfrag{9.5}[cc][cc][0.75]{}
\psfrag{10.5}[cc][cc][0.75]{}
\psfrag{11.5}[cc][cc][0.75]{}
\psfrag{12}[cc][lc][0.75]{$12$}
\psfrag{14}[cc][lc][0.75]{$14$}
\psfrag{16}[cc][lc][0.75]{$16$}
\psfrag{18}[cc][lc][0.75]{$18$}
\psfrag{x1}[cc][cc][0.75]{$600$}
\psfrag{x2}[cc][cc][0.75]{$900$}
\psfrag{x3}[cc][cc][0.75]{$1200$}
\psfrag{x4}[cc][cc][0.75]{$1500$}
\psfrag{x5}[cc][cc][0.75]{$1800$}
\psfrag{xlabel}[tc][Bc][0.85]{distance between BSs (in meters)}
\psfrag{ylabel}[Bc][tc][0.85]{{Average cell throughput (b/s/Hz)}}
\psfrag{ILM}[Bl][Bl][0.72]{ILM}
\psfrag{K=1, no IA}[Bl][Bl][0.75]{$K=1$, no IA}
\psfrag{K=2, no IA}[Bl][Bl][0.75]{$K=2$, no IA}
\psfrag{K=3, no IA}[Bl][Bl][0.75]{$K=3$, no IA}
\psfrag{K=4, no IA}[Bl][Bl][0.75]{$K=4$, no IA}
\psfrag{K=1, with IA, q=6}[Bl][Bl][0.75]{$K=1$, with IA, $q=6$}
\psfrag{K=2, with IA, q=2}[Bl][Bl][0.75]{$K=2$, with IA, $q=2$}
\psfrag{K=3, with IA, q=1}[Bl][Bl][0.75]{$K=3$, with IA, $q=1$}
\psfrag{K=4, with IA, q=0}[Bl][Bl][0.75]{$K=4$, with IA, $q=0$}
\includegraphics[width=3.4in]{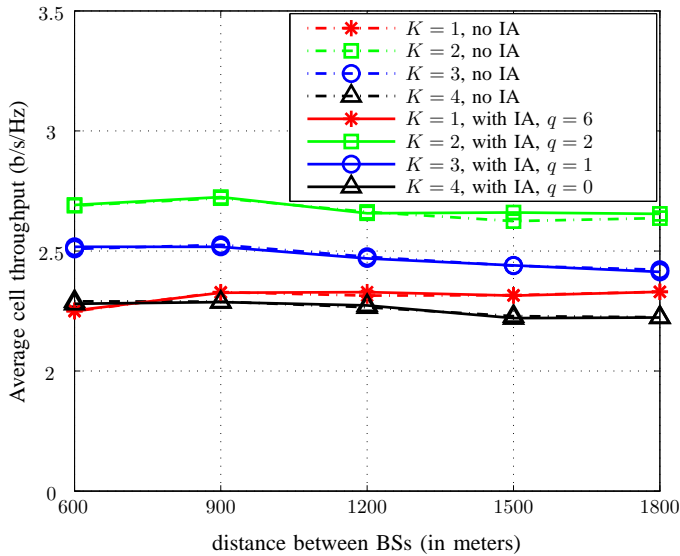}
\vspace{0mm}
	\caption{Average per-cell throughput in a $(7,K,4\times 4)$ network forming a 7-cell hexagonal cluster in a 49-cell hexagonal network when maximizing the minimum user rate under per-BS power constraints.}
	\label{C7RnonisominmaxA}
      \end{center}
    \end{figure}

Can IA impact wireless cellular network optimization? The evidence contained in this paper suggests that the impact of IA is quite limited even before accounting for the overhead and the required accuracy of CSI estimation. To arrive at this conclusion this paper first establishes certain fundamental results on feasibility of partial IA and uses these results to devise a two-stage optimization framework for NUM. The first stage of this framework focuses on interference nulling through partial IA followed by utility maximization in the second stage. The proposed framework is designed to leverage the strengths of IA and to overcome the shortcoming of conventional NUM algorithms. Through simulations on different cluster topologies with and without out-of-cluster interference, it is observed that IA is valuable in network topologies with a small number of BSs and without significant uncoordinated interference. In networks with significant out-of-cluster interference, nulling interference from a few dominant BSs does not appear to make an impact on the performance of NUM algorithms. Thus, in dense cellular networks, IA is likely to play a limited role even with centralized network optimization and full CSI.

\appendices

\section{Proof of Theorem \ref{feasibilitytheorem}}
\label{algebraicgeometryproof}
\subsection{Preliminaries}

The goal of this section is to present a concise introduction to the tools used in the proof of Theorem \ref{feasibilitytheorem}. Most of this material has been presented in various forms in earlier papers \cite{razaviyayn,annapureddy,tingtingliu} and is presented here for completeness and to bring more clarity to the concepts involved.

\subsubsection{Transcendental Field Extensions}

Let $\F$ be a field and $\F[x_1, \ldots, x_n]$ and $\F(x_1, \ldots, x_n)$ denote the ring of polynomials and rational functions over $\F$ respectively. Let $\K$ be a field containing $\F$ and denote the field extension by $\K/\F$. 
\begin{definition}
An element $\alpha \in \K$ is \emph{algebraic} over $\F$ if there exists a nonzero $f \in \F[x]$ such that $f(\alpha) = 0$. If no such $f$ exists, then $\alpha$ is \emph{transcendental} over $\F$. A set $S =\{\alpha_1, \ldots, \alpha_n \} \subset \K$ is \emph{algebraically dependent} over $\F$ if there exists a nonzero $f \in \F[x_1, \ldots, x_n]$ such that $f(\alpha_1, \ldots, \alpha_n) = 0$. Otherwise $S$ is \emph{algebraically independent} over $\F$. 
\end{definition}
Clearly, algebraic independent elements over $\F$ are transcendental over $\F$. 
Let $S =\{\alpha_1, \ldots, \alpha_n \} \subset \K$ be an algebraically independent set over $\F$. We can consider \emph{adjoining} the elements of $S$ to $\F$, denoted by $\F(S) = \F(\alpha_1, \ldots, \alpha_n)$. $\F(S)$ is defined to be the smallest field extension of $\F$ containing all elements of $S$. The following lemma shows that the field $\F(S)$ has an easy representation. 
\begin{lemma}
Let $\K/\F$ be a field extension. If $\alpha_1, \ldots, \alpha_n \in \K$ are algebraically independent over $\F$, then $\F(\alpha_1, \ldots, \alpha_n)$ and $\F(x_1, \ldots, x_n)$ are isomorphic (as field extensions of $\F$). 
\end{lemma}
\begin{definition}
A subset $S \subset \K$ is a \emph{transcendence basis} for $\K/\F$ if $S$ is algebraically independent over $\F$ and $\K$ is algebraic over $\F(S)$. 
\end{definition}
\begin{example}
Let $\K = \F(x_1, \ldots, x_n)$, then ${x_1, \ldots, x_n}$ is a transcendence basis for $\K/\F$. 
\end{example}
We should expect any two bases to have the same size, and this is indeed the case. We shall define this invariant. 
\begin{definition}
The \emph{transcendence degree} $\text{trdeg}(\K/\F)$ of a field extension $\K/\F$ is the cardinality of any transcendence basis of $\K/\F$. 
\end{definition}
The tools we developed so far gives us the following proposition, which we will use as a key step in the necessary part of Theorem \ref{feasibilitytheorem}.
\begin{proposition}
\label{main_nec}
Let $\K = \F(x_1, \ldots, x_n)$. Any set $S =\{\alpha_1, \ldots, \alpha_m \} \subset \K$ with $m > n$ is algebraically dependent over $\F$. 
\end{proposition}
\begin{proof}
Follows from the fact that $\text{trdeg}(\K/\F) = n < m$. 
\end{proof}

\subsubsection{Zariski Topology and a Theorem of Chevalley}:
Let $\K$ be an algebraically closed field (e.g. $\mathbb{C}$). Let $S \subset \K[x_1, \ldots, x_n]$ be a set of polynomials. Define the \emph{zero-locus} $Z(S)$ as:
\begin{align*}
Z(S) = \{x \in \K^n \mid f(x) = 0 \forall f \in S\}.
\end{align*}
A subset $V$ of $\K^n$ is called an \emph{affine algebraic set} if $V=Z(S)$ for some $S$. The \emph{Zariski topology} on $\K^n$ is defined by specifying the closed sets to be the affine algebraic sets. Thus, open sets are of the form $\K^n \setminus Z(S)$ for some $S \subset \K[x_1, \ldots, x_n]$. Intuitively, open sets are ``big'' in Zariski topology. This is made precise by the fact that open sets are dense (their closures are equal to $\K^n$). Zariski open sets allow us to define a property to be generic as follows.

\begin{definition}
A property of $\K^n$ is said to be true \emph{generically} if it is true over a non-empty Zariski open set of $\K^n$. 
\end{definition}
Closely related to open and closed sets is the concept of constructible sets. 
\begin{definition}
A set is \emph{locally closed} if it is the intersection of an open set with a closed set. A finite union of locally closed sets is called a \emph{constructible set}. 
\end{definition}
Two important facts related to constructible sets that are used in the proof are as follows.
\begin{proposition}
\label{opendense}
Every constructible set contains a dense open subset of its closure. 
\end{proposition}
\begin{theorem}[Special case of Chevalley Theorem]
Let $f_1, \dots, f_n \in \K[x_1, \ldots, x_n]$, and define $f = (f_1, \ldots, f_n) : \K^n \rightarrow \K^n$ to be the corresponding polynomial map. Then the image of $f$ ($\text{Im}(f)$) is a constructible set. 
\end{theorem}
A useful set of equivalent conditions that are satisfied by polynomial maps are presented in the following proposition. 
\begin{definition}
A polynomial map $f  = (f_1, \ldots, f_n): \K^n \rightarrow \K^n$ is \emph{dominant} if $\text{Im}(f)$ is dense in $\K^n$. 
\end{definition}
\begin{proposition}[\cite{elkadi}, Prop. 5.2]
\label{prop_jac}
For a polynomial map $f  = (f_1, \ldots, f_n): \K^n \rightarrow \K^n$, the following conditions are equivalent. 
\begin{enumerate}
\item $f$ is a dominant map. 
\item The function $f_1, \ldots, f_n$ are algebraically independent over $\K$. 
\item The Jacobian $J_f = \det\left (\left [\frac{\partial f_i}{\partial x_j}\right ]_{i,j} \right )$ of $f$ is not identically zero. 
\end{enumerate}
\end{proposition}
The above discussions give us the following proposition, which we will use as a key step in the sufficiency part of Theorem \ref{feasibilitytheorem}.
\begin{proposition}
\label{main_suff}
Let $f = (f_1, \ldots, f_n) : \K^n \rightarrow \K^n$ be a dominant polynomial map. Then $\text{Im}(f)$ contains a non-empty Zariski open set. 
\end{proposition}
\begin{proof}
By Chevalley's theorem, $\text{Im}(f)$ is contructible. Since $f$ is dominant, then the closure of $\text{Im}(f)$ is $\K^n$. By Proposition \ref{opendense}, $\text{Im}(f)$ contains a dense open subset of $\K^n$. 
\end{proof}

\subsection{Proof of Theorem \ref{feasibilitytheorem}}

This section proves a slightly more general form of Theorem \ref{feasibilitytheorem} where the $G$-cell network is permitted to have different number of users in each cell. Such networks are represented as $(G,\{K_g\},M\times N)$ networks. The new theorem statement follows.

\begin{theorem}
\label{newfeasibilitytheorem}
Consider a $(G,\{K_g\},M\times N)$ network where each user is served with one data stream. Let $\bt u_{gk}$ and $\bt v_{gk}$ denote the transmit and receive beamformer corresponding to the $(g,k)$th user where the set of beamformers $\{\bt u_{g1},\bt u_{g2},\hdots,\bt u_{gK_g} \}$ is linearly independent for every $g$. Further, let $\mathcal{I} \subseteq \{(i,gk):\, g\neq i,\, 1 \leq g,\, i\leq G,\, 1 \leq k \leq K_g \}$ be a set of BS-user pairs such that for each $(i,gk) \in \mathcal{I}$ the interference caused by the $i$th BS at the $(g,k)$th user is completely nulled, i.e.,
\begin{align}
\bt v^H_{gk} \bt H_{(i,gk)} \bt u_{ij}= 0, \qquad \forall j \in \{1,2,\hdots, K_i \}.
\end{align}
A set of transmit and receive beamformers $\{ \bt u_{gk} \}$ and $\{ \bt v_{gk}\}$ satisfying the polynomial system defined by $\mathcal{I}$ exist if and only if
\begin{align}
 M\geq & \ 1 \\
 N \geq  &\ K_g, \quad \forall  g.
\end{align}
and
\begin{align}
\label{eqmain}
|\mathcal{J}_{users}|(M-1)+\sum_{l\in \mathcal{J}_{BS}}(N-K_l )K_l  \geq \sum_{(l,gk)\in 
\mathcal{J}} K_l 
\end{align}
where $\mathcal{J}$ is any subset of $\mathcal{I}$ and $\mathcal{J}_{users}$ and $\mathcal{J}_{BS}$ are the set of user and BS indices that appear in $\mathcal{J}$.
\end{theorem}

\begin{proof}
The proof closely follows the proof presented in \cite{razaviyayn} to establish a similar feasibility result. 

Let the beamformers used by BS $i$ be collectively represented as the matrix $\bt U_i$, i.e., $\bt U_i=[\bt u_{i1},\bt u_{12},\hdots, \bt u_{iK}]$ . Let $\bt v_{gk}$ and $\bt U_i$ be such that $(i,gk) \in \J$. The IA condition implies that $\bt U_i$ must have rank $K_i$. Thus, we can apply invertible linear transformations to $\bt v_{gk}$ and $\bt U_i$ such that 
\begin{align*}
\bt v_{gk} = \bt P_{gk}^v \begin{bmatrix}
       1      \\
       \hvgk    
     \end{bmatrix} \bt R_{gk}^v  \quad  \bt U_{i} = \bt  V_{i}^u \begin{bmatrix}
       \bt I_{K_i \times K_i}      \\
       \hul   
     \end{bmatrix} \bt R_{i}^u, 
\end{align*}
where $\bt P_{gk}^v$ and $\bt P_{i}^u$ are square permutation matrices (consequently, their transpose equals their inverse) while $\bt R_{gk}^v$ and $\bt R_{i}^u$ are two invertible matrices. Defining $\hhgkl = {\bt P_{gk}^v}^{-1}  \bt H_{(i,gk)} {\bt P_{i}^u}^{-1}$, we partition it in the following way. 
\begin{align*}
\hhgkl = \begin{bmatrix}
       \hhgkla & \hhgklb      \\
       \hhgklc  & \hhgkld 
     \end{bmatrix},
\end{align*} 
where $\hhgkla$ has size $1 \times K_i$. Note that $\hhgkl$ is still a generic matrix. 
With the above transformation, we can rewrite the IA condition as 
\begin{align*}
\begin{bmatrix}
       1 & \hvgk^H      \\
     \end{bmatrix}
		\begin{bmatrix}
        \hhgkla & \hhgklb      \\
       \hhgklc  & \hhgkld 
     \end{bmatrix}
		\begin{bmatrix}
       \bt I    \\
      \hul
     \end{bmatrix} = \bt 0.
\end{align*}
This can be expanded as the following equation. 
\begin{align}
\label{heq}
\hhgkla + \hvgk^H\hhgklc + \hhgklb\hul + \hvgk^H\hhgkld\hul = \bt 0.
\end{align}

To establish the necessity part of the theorem, first note that the total number of scalar equations in (\ref{heq}) is 
\begin{align*}
\sum_{(i,gk)\in\J} K_i,
\end{align*}
and the total number of scalar variables (unknown entries in $\{\hvgk\}$'s and $\{\hul\}$'s) is 
\begin{align*}
|\mathcal{J}_{users}|(M-1)+\sum_{i\in \mathcal{J}_{BS}}(N-K_i )K_i. 
\end{align*}
Thus if 
\begin{align}
\label{revineq}
|\mathcal{J}_{users}|(M-1)+\sum_{i\in \mathcal{J}_{BS}}(N-K_i )K_i < \sum_{(i,gk)\in\J} K_i,
\end{align}
then we would have more equations than unknowns in $(\ref{heq})$. We show that no solution (for $\{\hvgk\}$'s and $\{\hul\}$'s) can exist in this case.

Consider a transcendental field extension $\F$ of $\mathbb{C}$ with a transcendence basis given by entries of $\{\hvgk, \hul\}_{(i,gk)\in \J}$. The transcendence degree of $\F$ is $|\mathcal{J}_{users}|(M-1)+\sum_{i\in \mathcal{J}_{BS}}(N-K_i )K_i$. 
Construct, for each $(i,gk) \in \J$, 
\begin{align}
\label{feq}
 \fgkl (\hvgk, \hul) = -\hvgk^H\hhgklc - \hhgklb\hul - \hvgk^H\hhgkld\hul. 
\end{align}
Note that $\fgkl$ is a $1 \times K_i$ vector with each entry in $\F$. In particular, each entry of $\fgkl$ is a quadratic polynomial of the entries in $\hvgk$ and $\hul$. If $(\ref{revineq})$ holds, then the total number of these entries (quadratic polynomials) in $\{\fgkl\}_{(i,gk)\in \J}$ is strictly greater than the transcendence degree of $\F$ over $\mathbb{C}$. Thus, by Proposition \ref{main_nec}, these entries are algebraically dependent over $\mathbb{C}$. In particular, there must exists a nonzero polynomial $p$ (in $\sum_{(i,gk)\in\J} K_i$ variables with coefficients in $\mathbb{C}$) such that 
\begin{align*}
p(\{\fgkl(\hvgk,\hul)\}_{(i,gk)\in \J}) = 0 \quad \forall \{\hvgk, \hul\}_{(i,gk)\in \J},  
\end{align*}
where the notation $p(\{\fgkl(\hvgk,\hul)\}_{(i,gk)\in \J})$ means that $p$ takes on each entry of every $\fgkl$ as an input in a specified order. Note that $p$ is independent of $\{\hhgkla\}_{(i,gk)\in \J}$. Thus if we view $p$ as a polynomial in the variable $X = (\{\hhgkla\}_{(i,gk)\in \J})$, then $p$ can be expanded locally at $\hat{X}= (\{\fgkl(\hvgk,\hul)\}_{(i,gk)\in \J})$ as
\begin{align}
&p(\{\hhgkla\}_{(i,gk)\in \J}) \nonumber \\
&\hspace{-0.1cm}=p(\{\fgkl(\hvgk,\hul)\}_{(i,gk)\in \J})\nonumber \\
&\hspace{0.2cm}+\hspace{-0.25cm}\sum_{(i,gk)\in \J}\hspace{-0.3cm}(\hhgkla - \fgkl(\hvgk,\hul))\bt{Q}_{i,gk}(\{\hhgkla\}_{(i,gk)\in \J}) \nonumber  \\
&\hspace{5cm}\forall \{\hvgk, \hul\}_{(i,gk)\in \J},
\end{align}
where $\bt{Q}_{i,gk}$ is some polynomial vector of size $K_i \times 1$. Our assumption on $p$ implies
\begin{align}
&p(\{\hhgkla\}_{(i,gk)\in \J})= \nonumber \\
&\sum_{(i,gk)\in \J}\hspace{-0.3cm}(\hhgkla - \fgkl(\hvgk,\hul))\bt{Q}_{i,gk}(\{\hhgkla\}_{(i,gk)\in \J})\nonumber \\
&\hspace{5cm}\forall \{\hvgk, \hul\}_{(i,gk)\in \J}, 
\end{align}
If equation $(\ref{heq})$ is satisfied, then there exists a choice of matrices $\{\hvgk, \hul\}_{(i,gk)\in \J}$ such that 
\begin{align}
\hhgkla - \fgkl(\hvgk,\hul) = 0 \quad \forall {(i,gk)\in \J}. 
\end{align}
For this choice, we have
\begin{align}
\label{contra}
p(\{\hhgkla\}_{(i,gk)\in \J}) = 0. 
\end{align}
However, $\{\hhgkla\}_{(i,gk)\in \J}$ is generic and independent of $p$. Thus ($\ref{contra}$) can only be satisfied if $p$ is identically the zero polynomial. This contradicts our assumption on $p$ and proves the necessity part of Theorem \ref{newfeasibilitytheorem}.

For sufficiency, we focus on the case when the total number of variables equals the total number of equations. All other cases follow easily. To establish the sufficiency part of the theorem, note that it suffices to find a choice of $\{\hhgklb,\hhgklc,\hhgkld\}_{(i,gk)\in \I}$ such that the Jacobian of the polynomial map ($\ref{feq}$) (in variables $\{ (\hvgk,\hul)\}_{(i,gk)\in \I}$) is nonzero. The condition that the Jacobian of a polynomial map is zero is an algebraic condition on $\{\hhgklb,\hhgklc,\hhgkld\}_{(i,gk)\in \I}$. Thus, if there exists a choice of $\{\hhgklb,\hhgklc,\hhgkld\}_{(i,gk)\in \I}$ such that the Jacobian is nonzero, then the Jacobian is nonzero for \emph{generic} choices of $\{\hhgklb,\hhgklc,\hhgkld\}_{(i,gk)\in \I}$. After establishing the Jacobian of the polynomial map is nonzero, Proposition \ref{prop_jac}, tells us that the map ($\ref{feq}$) is in fact dominant. Then Proposition \ref{main_suff} tells us that the image of the map ($\ref{feq}$) contains a non-empty Zariski open set $U$ of $\F$. Thus, equation (\ref{heq}) holds for all $(\hhgkla)_{(i,gk)\in \I} \in U$ and therefore holds generically. 

We now establish a choice of $\{\hhgklb,\hhgklc,\hhgkld\}_{(i,gk)\in \I}$ such that the Jacobian of the polynomial map ($\ref{feq}$) is nonzero. The construction of the Jacobian closely follows the construction presented in \cite{tingtingliu}.

Before constructing the Jacobian matrix, we create a single concatenated vector of variables by ordering the variables $\{\hul\}$ in a lexicographic manner followed by the variables $\{\hvgk\}$ also listed in a similar manner. A list of equations is created by first listing all equations (as given in (\ref{feq})) that involve interference cancellation from the first BS, followed by the second BS, and so on. Let this vectorized list of variables and equations be denoted as  $\bs  \lambda$ and $\bs \psi$ respectively. Note that both vectors are of length $\sum_{i=1}^G \left ( (N-K_i)K_i + (M-1)K_i \right )$. The part of $\bs \lambda$ that corresponds to the $\{ \hul \}$ variables is denoted as $\bs \lambda_{\bar{\bt u}}$. Similarly define $\bs \lambda_{\bar{\bt v}}$. The part of $\bs \psi$ that corresponds to equations involving $\hul$ is denoted as $\bs \psi_{\hul}$. Further, let the number of equations that involve the $i$th BS's beamformers be given by $e_i$.

The $(i,j)$th entry in the Jacobian matrix $\bt J$ is given by $\tfrac{\partial \bs \psi_i } {\partial \bs \lambda_j}$. The notations $\tfrac{\partial \bs \psi_{\hul}}{\partial \hul}$, $\tfrac{\partial \bs \psi_{\huij}}{\partial \huij}$, $\tfrac{\partial \bs \psi_{\huij}}{\partial \lambda_{\bar{\bt v}}}$ all refer to submatrices of $\bt J$ are straightforward to infer.

In the Jacobian matrix we construct, we set $\{ \hhgkld \}$ to zero for all $g$, $k$, and $i$. We are left with choosing values for the $\{ \hhgklb \}$ and $\{ \hhgklc \}$ matrices. The structure of the resulting Jacobian matrix is illustrated using the following example.

\setlength{\arraycolsep}{0.02cm}
\setcounter{MaxMatrixCols}{20}
\begin{figure*}[t]
 \begin{align}\small
 \label{Jmatrix} 
  \begin{bmatrix}
   \makebox[0pt][l]{\hspace{0.05cm}$\smash{\overbrace{\phantom{H^j}\hspace{1.5cm}}^{\frac{\partial \bs \psi}{\partial \bar{\bt U}_1}}}$}
   \bar{\bt H}^{(2)}_{(1,21)} & 0 & 0 & 0 & \makebox[0pt][l]{\hspace{-0.23cm}$\smash{\overbrace{\phantom{H^j}\hspace{0.0cm}}^{\frac{\partial \bs \psi}{\partial \bar{\bt u}_{31}}}}$} 0 & 0 & 0 & 0 & \bar{\bt H}^{(3)}_{(1,21)11} & 0 & 0 & 0  \\
   \bar{\bt H}^{(2)}_{(1,32)} & 0 & 0 & 0 & 0 & 0 & 0 & 0 & 0 & 0 & 0 & \bar{\bt H}^{(3)}_{(1,32)11}  \\
   0 & \bar{\bt H}^{(2)}_{(1,21)}  & 0 & 0 & 0 & 0 & 0 & 0 & \bar{\bt H}^{(3)}_{(1,21)12} & 0 & 0 & 0  \\
   0 & \bar{\bt H}^{(2)}_{(1,32)}  & 0 & 0 & 0 & 0 & 0 & 0 & 0 & 0 & 0 & \bar{\bt H}^{(3)}_{(1,32)12}  \\
   0 & 0 & \bar{\bt H}^{(2)}_{(2,11)} & 0 & 0 & 0 & \bar{\bt H}^{(3)}_{(2,11)11} & 0 & 0 & 0 & 0 & 0   \\
   0 & 0 & \bar{\bt H}^{(2)}_{(2,31)} & 0 & 0 & 0 & 0 & 0 & 0 & 0 & \bar{\bt H}^{(3)}_{(2,31)11} & 0   \\
   0 & 0 & 0 & \bar{\bt H}^{(2)}_{(2,11)}  & 0 & 0 & \bar{\bt H}^{(3)}_{(2,11)12} & 0 & 0 & 0 & 0 & 0  \\
   0 & 0 & 0 & \bar{\bt H}^{(2)}_{(2,31)}   & 0 & 0 & 0 & 0 & 0 & 0 & \bar{\bt H}^{(3)}_{(2,31)12} & 0  \\
   0 & 0 & 0 & 0 & \bar{\bt H}^{(2)}_{(2,11)}  & 0 & 0 & \bar{\bt H}^{(3)}_{(3,12)11} & 0 & 0 & 0 & 0  \\
   0 & 0 & 0 & 0 & \bar{\bt H}^{(2)}_{(2,31)}  & 0 & 0 & 0 & 0 & \bar{\bt H}^{(3)}_{(3,22)11} & 0 & 0  \\
   0 & 0 & 0 & 0 & 0 & \bar{\bt H}^{(2)}_{(2,11)}   & 0 & \bar{\bt H}^{(3)}_{(3,12)12} & 0 & 0 & 0 & 0  \\
   \makebox[0pt][l]{\hspace{-0.2cm}$\smash{\underbrace{\phantom{H_j}\hspace{5.6cm}}_{\frac{\partial \bs \psi}{\partial {\bs \lambda}_{\bar{\bt u}}}}}$} 0 & 0 & 0 & 0 & 0 & \bar{\bt H}^{(2)}_{(2,31)}   & \makebox[4pt][l]{\hspace{-0.2cm}$\smash{\underbrace{\phantom{H_j}\hspace{6.9cm}}_{\frac{\partial \bs \psi}{\partial {\bs \lambda}_{\bar{\bt v}}}}}$} 0 & 0 & 0 & \bar{\bt H}^{(3)}_{(3,22)12} & 0 & 0  \\
   \end{bmatrix}
   \begin{matrix}
     \bigg \} {\frac{\partial \bs \psi_{\bar{\bt u}_{11}}}{\partial  \bs  \lambda}}  \\ \phantom{\Big \}} \\  \phantom{\Bigg \}}  \\   \phantom{\Big \}} \\ \phantom{\Big \}}\\ \phantom{} \\ \left \}\rule{0cm}{1.03cm}  {\frac{\partial \bs \psi_{\bar{\bt U}_{3}}}{\partial  \bs  \lambda}} \right. \\
    \end{matrix}
  \end{align}
  \end{figure*}

\begin{example}	
\label{Jeg}
 Consider the $(3,2,2\times3)$ network with the set $\mathcal I$ given by  $\{ (1,21),\allowbreak (1,32),\allowbreak (2,11),\allowbreak (2,31),\allowbreak (3,12),\allowbreak (3,22) \}$. Then $\bs \lambda=[\bar{\bs U}_{111},\allowbreak \bar{\bs U}_{112},\allowbreak \bar{\bs U}_{211}, \bar{\bs U}_{212},\allowbreak \bar{\bs U}_{311}, \bar{\bs U}_{312},\allowbreak \bar{\bs v}_{1111}, \bar{\bs v}_{1211},\allowbreak \bar{\bs v}_{2111},\allowbreak \bar{\bs v}_{2211},\allowbreak \bar{\bs v}_{3111},\allowbreak \bar{\bs v}_{3211}  ]$ and $\bs \psi$ is given by $[ \bt F_{(1,21)1},\bt F_{(1,32)1},\allowbreak \bt F_{(1,21)2}, \allowbreak \bt F_{(1,32)2},\allowbreak \bt F_{(2,11)1}, \allowbreak \bt F_{(2,31)1}, \allowbreak \bt F_{(2,11)2},\allowbreak \bt F_{(2,31)2}, \allowbreak \bt F_{(3,12)1}, \allowbreak \bt F_{(3,22)1},\allowbreak \bt F_{(3,12)2},\allowbreak \bt F_{(3,22)2}]$, where $\bar{\bt U}_{ipq}$ refers to the $(p,q)$th element of $\bar{\bt U}_i$, $\bar{\bt v}_{ijpq}$ refers to the $(p,q)$th element of $\bar{\bt v}_{ij}$ and ${\bt F}_{(i,gk)p}$ refers to the $p$th equation of ${\bt F}_{(i,gk)}$. The $12\times 12$ Jacobian matrix and the various submatrices are presented in (\ref{Jmatrix}). In this example, each $\bar{\bt H}_{(i,gk)}^{(2)}$ matrix is a $1\times 1$ matrix, while each $\bar{\bt H}_{(i,gk)}^{(3)}$ matrix is a $1\times 2$ matrix. Note the block diagonal structure on the left with each block repeated twice.
\end{example}

\begin{figure*} 
\centering
 \subfloat[]{
                \includegraphics[width=1.2in]{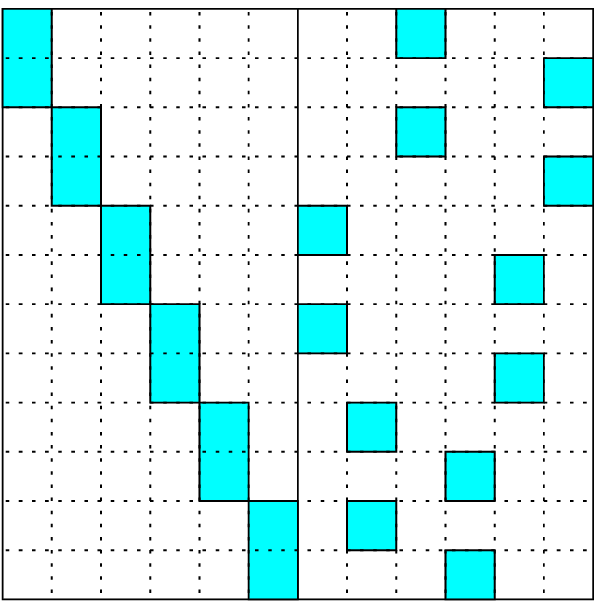}
		\label{structurefiga}
        }\hfil
	\subfloat[]{
	        \includegraphics[width=1.2in]{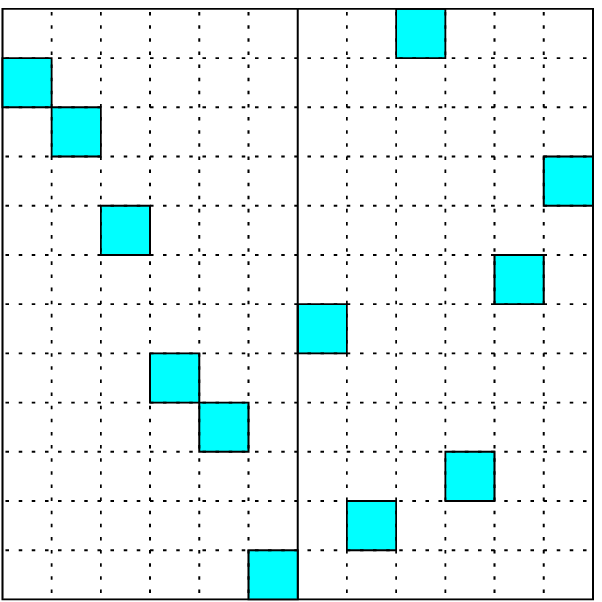}
		\label{structurefigb}
        }\hfil
        \subfloat[]{
                \includegraphics[width=1.2in]{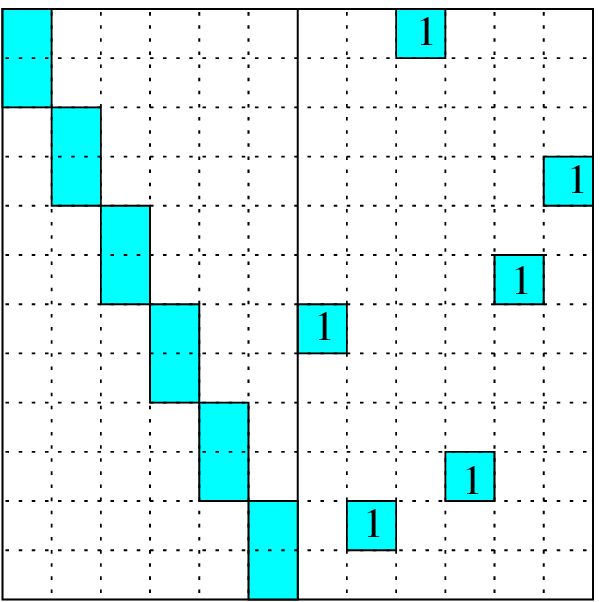}
		\label{structurefigc}
        }%
	\caption{Structure of (a) the original Jacobian matrix; (b) after reduction to a permutation matrix; and (c) final structure of the Jacobian matrix after preserving the permutation structure only on the right half, for the $(3,2,2\times 3)$ network in Example \ref{Jeg}.}\label{structurefig}
\end{figure*}

Setting $\{ \hhgkld \}$ to zero results in a Jacobian matrix that has a repeating structure on the left (corresponding to the submatrices $\tfrac{\partial \bs \psi_{\huij}}{\partial \huij}$ that have the same partial derivatives for a fixed $i$ as $j$ is varied and a sparse structure on the right. Note that no channel elements get repeated in the $\tfrac{\partial \bs \psi}{\partial \bs \lambda_{\bar{\bt v}}}$ submatrix and that two different sets of channel elements are involved in the $\tfrac{\partial \bs \psi}{\partial \bs \lambda_{\bar{\bt u}}}$ and $\tfrac{\partial \bs \psi}{\partial \bs \lambda_{\bar{\bt v}}}$ submatrices. Further note that each row of such a Jacobian matrix has as many non-zero entries as the number of variables involved in the equation corresponding to that row (see (\ref{feq})).

The structure of such a matrix can be represented as a bipartite graph where two vertices of such a bipartite graph are connected by an edge if the corresponding element in the matrix is non-zero. When the necessary condition (\ref{eqmain}) holds, it can be shown that the bipartite graph constructed through the Jacobian matrix in the above manner satisfies the necessary conditions for Hall's theorem \cite{diestel} which guarantees the existence of a perfect matching in such a matrix. Existence of perfect matching in such a  graph is equivalent to the ability to reduce the Jacobian matrix to a permutation matrix by setting certain channel values to zero (while ignoring the repetitions of certain channel values and treating all entries in the Jacobian matrix to be independent of each other). Fig.~\ref{structurefiga} represents the structure of the Jacobian matrix for the example discussed earlier with gray cells representing non-zero entries and Fig.~\ref{structurefigb} is one possible permutation matrix that such a matrix can be reduced to. 

We retain the permutation structure that results from such a reduction only in the $\tfrac{\partial \bs \psi}{\partial \lambda_{\bar{\bt v}}}$ submatrix of $\bt J$ (right half of Fig.~\ref{structurefigb}) by setting all non-zero channel values to 1. The structure of the resulting Jacobian matrix for the example discussed earlier is given in Fig.~\ref{structurefigc}. With the $\tfrac{\partial \bs \psi}{\partial \lambda_{\bar{\bt v}}}$ fixed in the above manner, it can be shown that for any random full-rank choice of the  submatrices $\tfrac{\partial \bs \psi_{\bar{\bt u}_{ij}}}{\partial {\bar{\bt u}_{ij}}}$ (while ensuring $\tfrac{\partial \bs \psi_{\bar{\bt u}_{ij}}}{\partial {\bar{\bt u}_{ij}}}=\tfrac{\partial \bs \psi_{\bar{\bt u}_{ik}}}{\partial {\bar{\bt u}_{ik}}}$ for $j\neq k )$, the resulting Jacobian is full-rank. To see this, first note that the rank of the Jacobian is now a sum of the ranks of the individual submatrices $\tfrac{\partial \bs \psi_{\bar{\bt u}_{ij}}}{\partial \bs \lambda}$. Note that each such submatrix has non-zero entries in mutually exclusive columns. Now, each submatrix $\tfrac{\partial \bs \psi_{\bar{\bt u}_{ij}}}{\partial \bs \lambda}$ has $e_i$ rows. By construction, the right side of this submatrix (this is the $\tfrac{\partial \bs \psi_{\bar{\bt u}_{ij}}}{\partial \lambda_{\bar{\bt v}}}$ submatrix) has $(e_i-N+K_i)$ ones on distinct rows. This is because, the submatrix $\tfrac{\partial \bs \psi_{\bar{\bt u}_{ij}}}{\partial {\bar{\bt u}_{ij}}}$ has $N-K_i$ columns and each column must have at least one entry chosen while constructing the permutation matrix, leaving $(e_i-N+K_i)$ non-zero entries in the $\tfrac{\partial \bs \psi_{\bar{\bt u}_{ij}}}{\partial \lambda_{\bar{\bt v}}}$ submatrix after adopting the permutation structure.

Using a column transformation and eliminating non-zero entries on $(e_i-N+K_i)$ rows on the left side (i.e., the $\tfrac{\partial \bs \psi_{\bar{\bt u}_{ij}}}{\partial \lambda_{\bar{\bt u}}}$ submatrix), we are left with exactly $(N-K_i)$ non-zero rows in $\tfrac{\partial \bs \psi_{\bar{\bt u}_{ij}}}{\partial \lambda_{\bar{\bt u}}}$ having non-zero entries in the same number of columns. It is easy to see that for any random full-rank choice of the channel matrices $\{ \hhgklb \}$, such a submatrix is full rank, thus proving that each of the submatrices is full-rank and hence the Jacobian has a non-zero determinant. Such a construction completes the proof of sufficiency.
\end{proof}

\bibliographystyle{IEEEtran}
\bibliography{IEEEabrv,ref_file}

\begin{IEEEbiographynophoto}{Gokul Sridharan} (S'08-M'15)
received his B.Tech. and M.Tech. (dual degree) in Electrical Engineering from the Indian Institute of Technology Madras, Chennai, India, in 2008, and the M.A.Sc. and Ph.D. degrees in Electrical Engineering from the University of Toronto, Toronto, Canada, in 2010 and 2014, respectively. He is currently a post-doctoral researcher at WINLAB, Rutgers, The State University of New Jersey, New Brunswick, New Jersey, USA. His research interests include wireless communications, convex optimization and information theory.
\end{IEEEbiographynophoto}

\begin{IEEEbiographynophoto}{Siyu Liu} (S'07-M'16)
Siyu Liu received the B.A.Sc degree in Electrical Engineering and the B.Sc degree in Mathematics and Physics concurrently, the M.A.Sc degree in Electrical Engineering, the M.Sc degree in Mathematics, and Ph.D. in Electrical Engineering all from the University of Toronto, ON, Canada, in 2007, 2009, 2010, and 2016 respectively. His research interests include coding theory, information theory, applied algebra and algebraic geometry.
\end{IEEEbiographynophoto}

\begin{IEEEbiographynophoto}{Wei Yu}(S'97-M'02-SM'08-F'14)
received the B.A.Sc. degree in Computer Engineering and Mathematics from the University of Waterloo, Waterloo, Ontario, Canada in 1997 and M.S. and Ph.D. degrees in Electrical Engineering from Stanford University, Stanford, CA, in 1998 and 2002, respectively. Since 2002, he has been with the Electrical and Computer Engineering Department at the University of Toronto, Toronto, Ontario, Canada, where he is now Professor and holds a Canada Research Chair (Tier 1) in Information Theory and Wireless Communications. His main research interests include information theory, optimization, wireless communications and broadband access networks.

Prof. Wei Yu currently serves on the IEEE Information Theory Society Board of Governors (2015-17). He is an IEEE Communications Society Distinguished Lecturer (2015-16). He served as an Associate Editor for IEEE Transactions on Information Theory (2010-2013), as an Editor for IEEE Transactions on Communications (2009-2011), as an Editor for IEEE Transactions on Wireless Communications (2004-2007), and as a Guest Editor for a number of special issues for the IEEE Journal on Selected Areas in Communications and the EURASIP Journal on Applied Signal Processing. He was a Technical Program co-chair of the IEEE Communication Theory Workshop in 2014, and a Technical Program Committee co-chair of the Communication Theory Symposium at the IEEE International Conference on Communications (ICC) in 2012. He was a member of the Signal Processing for Communications and Networking Technical Committee of the IEEE Signal Processing Society (2008-2013). Prof. Wei Yu received a Steacie Memorial Fellowship in 2015, an IEEE Communications Society Best Tutorial Paper Award in 2015, an IEEE ICC Best Paper Award in 2013, an IEEE Signal Processing Society Best Paper Award in 2008, the McCharles Prize for Early Career Research Distinction in 2008, the Early Career Teaching Award from the Faculty of Applied Science and Engineering, University of Toronto in 2007, and an Early Researcher Award from Ontario in 2006. Prof. Wei Yu was named a Highly Cited Researcher by Thomson Reuters in 2014. He is a registered Professional Engineer in Ontario.
\end{IEEEbiographynophoto}

\end{document}